\documentclass[pdflatex,sn-nature]{sn-jnl}

\usepackage{amsmath,amssymb,amsfonts}
\usepackage{amsthm}
\usepackage{graphicx}
\usepackage{xcolor}
\usepackage{booktabs}
\usepackage{mathtools}
\usepackage{physics}
\usepackage{manyfoot}

\theoremstyle{plain}
\newtheorem{theorem}{Theorem}
\newtheorem{proposition}[theorem]{Proposition}
\newtheorem{corollary}[theorem]{Corollary}

\theoremstyle{remark}
\newtheorem{lemma}[theorem]{Lemma}
\newtheorem{remark}[theorem]{Remark}

\theoremstyle{definition}
\newtheorem{definition}[theorem]{Definition}

\newcommand{\id}{\mathbb{I}}
\newtheorem{supptheorem}{Supplementary Theorem}
\newtheorem{suppproposition}[supptheorem]{Supplementary Proposition}
\newtheorem{supplemma}[supptheorem]{Supplementary Lemma}
\newtheorem{suppcorollary}[supptheorem]{Supplementary Corollary}
\theoremstyle{definition}

\theoremstyle{remark}
\newtheorem{suppremark}[supptheorem]{Supplementary Remark}

\newcommand{\cH}{\mathcal{H}}
\newcommand{\cT}{\mathcal{T}}
\newcommand{\cX}{\mathcal{X}}

\newcommand{\eps}{\varepsilon}

\newcommand{\Aauth}{A_{12}}
\newcommand{\Ccoll}{C_{13}}

\newcommand{\Qfin}{\mathsf Q_{\rm fin}}   
\newcommand{\Qd}{\mathsf Q_{\le d}}       
\newcommand{\Qcl}{\mathsf Q_{\rm cl}}     
\newcommand{\Qco}{\mathsf Q_{\rm co}}     

\begin{document}

\title[Strategic non-shareability of quantum correlations]{Strategic Non-Shareability of Quantum Correlations}

\author*[1,2]{\fnm{Fumin} \sur{Wang}}

\affil*[1]{\orgdiv{MED-X Institute}, \orgname{The First Affiliated Hospital of Xi'an Jiaotong University}, \orgaddress{\city{Xi'an}, \postcode{710061}, \country{China}}}

\affil[2]{\orgdiv{Shaanxi Key Laboratory of Quantum Information and Quantum Optoelectronic Devices, College of Physics}, \orgname{Xi'an Jiaotong University}, \orgaddress{\city{Xi'an}, \postcode{710049}, \country{China}}}

\abstract{Correlations distributed by a mediator can be useful for
coordination but vulnerable to inheritance by a colluder. We formalize
the obstruction to such inheritance as a source-certified resource theory
of strategic non-shareability. The free objects are symmetrically
extendible sources, the free operations are shareability-preserving maps,
and the trace distance to the free set is a faithful convex monotone.

For Werner and isotropic sources in arbitrary local dimension, the
resource has the exact form $D_m=c(d)(p-p_m^{*})_{+}$, with $p_m^{*}$
the Johnson--Viola shareability threshold. For qubit Werner sources,
tomographically complete Pauli measurements yield the exact one-colluder
capacity
\[
C^{\rm tomo}_1(p)=\frac{1}{12}\Bigl[(3p-1)-\sqrt{(3p+1)(1-p)}\,\Bigr]_{+}.
\]

We prove that this anti-collusion resource is independent of Bell
nonlocality: the Bell and shareability orderings cross, so some
Bell-nonlocal states are strictly less collusion-resistant than
Bell-local ones. Finally, we give an aligned Pauli coordination game
whose observed behaviour has a local hidden-variable model for every
visibility, making device-independent certification empty, while
source-certified quantum anti-collusion is positive exactly above the
extendibility threshold. These results identify symmetric
non-extendibility, rather than Bell nonlocality, as the boundary of
source-certified collusion resistance.}


\maketitle

\section{Introduction}\label{sec:intro}

Bell nonlocality is a foundational resource for quantum information processing~\cite{Bell1964,CHSH1969,Tsirelson1980,Brunner2014}. Device-independent protocols leverage it to guarantee security without trusting devices~\cite{Barrett2005,ArnonFriedman2018,Portmann2022,Pironio2009,Acin2007,Masanes2011,Ekert1991}, with recent experiments achieving device-independent quantum key distribution over long distances~\cite{Zapatero2023,Nadlinger2022,Zhang2022,Liu2022} and motivating quantum network architectures~\cite{Wehner2018,Tavakoli2022}. This security rests on the monogamy of Bell correlations~\cite{Terhal2004,Toner2006,Coffman2000,Osborne2006,Scarani2001}: the same strong correlation cannot be shared with a third party.

Beyond cryptography, quantum correlations mediate strategic games, enlarging achievable equilibria beyond classical limits~\cite{Aumann1974,Cleve2004,Zhang2012,Pappa2015}. Classical mediators rely on shared randomness, which is freely copyable and inherently vulnerable to collusion. Quantum mediators are constrained by no-cloning~\cite{Wootters1982} and entanglement monogamy~\cite{Horodecki2009,Werner1989}. Yet a rigorous framework that quantifies this defense---and certifies it from observed data---has been missing.

The key point is that collusion resistance is not a property of a behavior alone, but of an extension model. In a broadcast-enabled behavior model, every Bell-local behavior is freely shareable: the hidden variable explaining the authorized correlation can simply be copied to the colluder. Thus device-independent certification is necessarily empty below the Bell boundary. In a quantum-source model, however, the colluder is not handed an external classical seed; it must be embedded in a genuine quantum extension of the mediator's source. The relevant obstruction is then symmetric non-extendibility, not Bell nonlocality. This distinction lets the same observed local statistics be useless for device-independent security but useful for source-certified collusion resistance.

Building on the state-level resource theory of unextendibility~\cite{KDWW2019,KDWW2021,WangWilde2024,SinghWilde2025}, we close this gap by formalizing strategic non-shareability as a quantum resource. The free objects are symmetrically extendible sources, the free operations are shareability-preserving maps, and the trace distance to the free set is a faithful convex monotone. For Werner and isotropic sources the resource is exactly solvable in arbitrary dimension~\cite{JV2013,JSZ2022}. Most importantly, the resource ordering is independent of Bell nonlocality: the two orderings strictly cross.

The resource has a direct operational interpretation. In an aligned Pauli coordination game, the observed Werner behaviour admits an explicit local hidden-variable model for every visibility, making device-independent certification empty. Nevertheless, above the source extendibility threshold every quantum-source colluder is bounded away from the authorized score. The same data are therefore useless for device-independent certification but useful for source-certified collusion resistance.

This shifts the role of monogamy from a Bell-violation witness to a source-certified constraint on correlation inheritance, relevant to quantum networks, mediated coordination, and semi-device-independent security. Our contributions are fourfold.
\begin{itemize}
\item[(i)] We give an operational aligned-coordination primitive: for Werner states the device-independent advantage is identically zero, yet above the source extendibility threshold every quantum-source colluder is bounded away from the authorized score, making source-certified collusion resistance strictly positive while broadcast/DI certification is empty.
\item[(ii)] We prove a strict ordering reversal: Bell nonlocality and source-certified anti-collusion are not merely separated by thresholds but induce genuinely incompatible orderings on quantum states.
\item[(iii)] We lift the resource theory of unextendibility to the behavior/strategic layer, prove monotonicity of the trace-distance resource, and show that the induced behavior-level capacity is the operational monotone observable in games.
\item[(iv)] We give exact source-resource formulas for Werner and isotropic families, an exact qubit tomographic capacity, and multi-round finite-data and collective-colluder guarantees.
\end{itemize}

\section{Results}\label{sec:results}

\subsection{Extension models and the collusive shadow}\label{sec:shadow}

We consider three agents. Players 1 and 2 form the authorized pair, and player 3 is a potential colluder. Each player $i$ receives a private type $t_i\in\cT_i$ and produces an action $x_i\in\cX_i$. After receiving types, players cannot communicate. A mediator supplies local recommendations, which may be classical or quantum.

A classical mediator is a local hidden-variable device,
\begin{equation}
P_{12}(x_1,x_2|t_1,t_2)
=
\sum_\lambda p(\lambda)
p_1(x_1|t_1,\lambda)
p_2(x_2|t_2,\lambda),
\label{eq:classical-mediator}
\end{equation}
and a quantum mediator uses a bipartite state and local POVMs,
\begin{equation}
P_{12}(x_1,x_2|t_1,t_2)
=
\Tr\!\left[
\rho_{12}
\left(E^{(1)}_{t_1,x_1}\otimes E^{(2)}_{t_2,x_2}\right)
\right].
\label{eq:quantum-mediator}
\end{equation}

The mediator is required to be no-signalling: for every subset $S\subseteq\{1,2,3\}$,
\begin{equation}
P(x_S|t_1,t_2,t_3)=P(x_S|t_S).
\label{eq:no-signalling-full}
\end{equation}
This condition rules out hidden communication after types are received. It does not forbid copied hidden variables; copied seeds are precisely the broadcast-enabled free model considered below.

A collusive extension is a tripartite distribution satisfying
\begin{equation}
\sum_{x_3}P_{123}(x_1,x_2,x_3|t_1,t_2,t_3)
=
P_{12}(x_1,x_2|t_1,t_2).
\label{eq:extension}
\end{equation}
The physical question is whether such an extension can preserve the authorized marginal while also giving the colluder a high score with player 1.

For an authorized behavior $P_{12}$, we define the \emph{collusive extension set} under resource class
\[
\mathsf R\in\{\mathsf C,\mathsf{NS},\Qfin,\Qd,\Qcl,\Qco\}
\]
as
\begin{equation}
{\rm Ext}_{\mathsf R}(P_{12})
=
\left\{
P_{123}\in\mathsf R:
\sum_{x_3}P_{123}(\cdot,\cdot,x_3|\cdot,\cdot,\cdot)
=
P_{12}
\right\},
\end{equation}
where $\mathsf C$ denotes classical hidden-variable extensions, $\mathsf{NS}$ general no-signalling extensions, $\Qfin$ finite-dimensional quantum extensions, $\Qd$ fixed-dimensional quantum extensions with local Hilbert-space dimension bounded by $d$, $\Qcl$ the closure of $\Qfin$, and $\Qco$ the commuting-operator model associated with the NPA limit. The closure $\Qcl$ is needed for compactness in Theorem~\ref{thm:distance}; for statements where the optimum is explicitly attained in $\Qfin$ itself, no closure is needed. Unless closure issues are relevant, we use ``quantum extension'' informally.

\begin{definition}[Collusive shadow]
Let $\Pi_{13\to12}$ denote the relabelling map that identifies player 3's type and action alphabets with player 2's. The \emph{collusive shadow} of an authorized behavior $P_{12}$ under resource class $\mathsf R$ is
\begin{equation}
\mathcal S_{\mathsf R}(P_{12})
=
\left\{
\Pi_{13\to12}P_{13}:
P_{123}\in{\rm Ext}_{\mathsf R}(P_{12})
\right\},
\end{equation}
where $P_{13}(x_1,x_3|t_1,t_3)=\sum_{x_2}P_{123}(x_1,x_2,x_3|t_1,t_2,t_3)$ and no-signalling ensures this marginal does not depend on $t_2$.
\end{definition}

The collusive shadow is the set of all authorized-looking behaviors that a colluder can cast with player 1 without disturbing the observed authorized marginal $P_{12}$. A mediator is strategically non-shareable when its own authorized behavior lies outside, or far from, this shadow.

\begin{figure}[t]
\centering
\includegraphics[width=\linewidth]{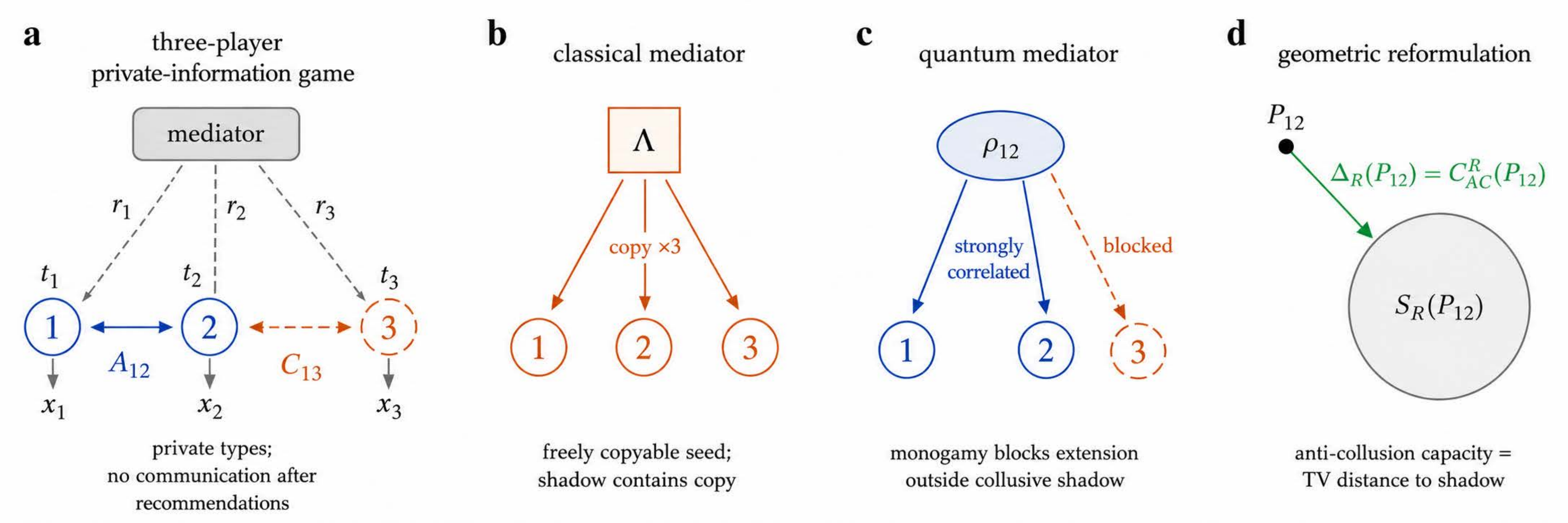}
\caption{\textbf{Strategic non-shareability of mediated correlations.} \textbf{a}, Three-player private-information game. Players 1 and 2 form the authorized pair; player 3 is a potential colluder. The mediator distributes private recommendations $r_i$ conditional on private types $t_i$, and players cannot communicate after receiving recommendations. The authorized score $A_{12}$ rewards the authorized pair, while the collusive leakage score $C_{13}$ measures the colluder's correlation with player 1. \textbf{b}, A classical mediator implements correlations via a shared hidden variable $\Lambda$, which can be freely copied to player 3; the collusive shadow $\mathcal S_C(P_{12})$ therefore contains a relabelled copy of the authorized behavior. \textbf{c}, A quantum mediator uses an entangled state $\rho_{12}$. Entanglement monogamy prevents the same pairwise correlation from being shared losslessly with player 3, so the authorized behavior can lie outside its collusive shadow. \textbf{d}, Geometric reformulation. The anti-collusion capacity is the total-variation distance from $P_{12}$ to the collusive shadow $\mathcal S_R(P_{12})$.}
\label{fig:copyable-monogamy}
\end{figure}

These quantities are linked by the following exact equality.

\begin{theorem}[Anti-collusion capacity equals distance from the collusive shadow]\label{thm:distance}
Let $P_{12}$ be an authorized behavior with finite input and output alphabets. Let
\[
\mathsf R\in\{\mathsf C,\mathsf{NS},\Qd,\Qcl\}
\]
be an extension class for which ${\rm Ext}_{\mathsf R}(P_{12})$ is non-empty, compact and convex. Then
\begin{equation}
\mathcal C_{\rm AC}^{\mathsf R}(P_{12})
=
\inf_{Q\in\mathcal S_{\mathsf R}(P_{12})}
d_{\rm TV}(P_{12},Q).
\label{eq:distance-theorem}
\end{equation}
\end{theorem}

The proof is a finite-dimensional minimax argument and is given in Supplementary Note~3. Here $d_{\rm TV}$ on behaviors denotes the input-weighted total-variation distance; the variational quantity $\sup_h\langle h,P-Q\rangle$ equals $\sum_t\mu(t)\cdot\frac12\sum_x|P(x|t)-Q(x|t)|$ with $\mu$ the input measure carried by the kernel. Thus a fixed game provides a separating hyperplane that witnesses non-shareability, whereas the game-optimized capacity equals the full operational distance from the collusive shadow.

\subsection{Separation theorem: device independence is empty, source certification is positive}\label{sec:separation}

%
%

The central operational result is a strict separation between device-independent and source-certified collusion resistance on the same observed behaviour. We exhibit an explicit family for which the DI side is identically zero while the source-certified side is strictly positive with an exact closed form.

Consider the game $G_{\mathrm{aligned}}$: a referee draws $i\in\{X,Y,Z\}$ uniformly and
sends the \emph{same} $i$ to Alice and to her partner; both measure $\sigma_i$ and score
$1$ iff their outcomes are anti-correlated. For a state $\rho$ the score is
$\mathcal S(\rho)=\tfrac13\sum_i\tfrac12\bigl(1-\langle\sigma_i\!\otimes\!\sigma_i\rangle_\rho\bigr)$.

\begin{proposition}[Explicit local model; vacuous device-independence]
\label{prop:c1-local}
The aligned Werner behaviour $P_{AB}(a,b\mid i,j)=\tfrac14\bigl(1-p\,\delta_{ij}\,ab\bigr)$
is reproduced exactly, for every $p\in[0,1]$, by the hidden-variable model: draw
$s\in\{\pm1\}^3$ uniformly; with probability $p$ Alice outputs $s_i$ and the partner
outputs $-s_i$, otherwise Alice outputs $s_i$ and the partner outputs an independent
uniform bit. Consequently the behaviour is local for all $p$, and in the broadcast model
a colluder holding $s$ reproduces it verbatim, achieving the authorised score with zero
anti-collusion power. No device-independent certificate (which requires Bell violation)
yields any guarantee here.
\end{proposition}

The explicit local model is the non-circularity anchor: it constructively shows that the broadcast and device-independent sides are identically zero, so any positive advantage established below is attributable solely to the source assumption.

\begin{theorem}[Source-certified coordination advantage]
\label{thm:c1}
In $G_{\mathrm{aligned}}$ the authorised pair sharing $\rho_W(p)$ scores
$A(p)=\tfrac{1+p}{2}$, while every quantum-source colluder (any tripartite extension with
$\rho_{AB}=\rho_W(p)$) scores at most $V_Q(p)=\tfrac{1+q^\star(p)}{2}$, attained by the
Werner colluder $\rho_W(q^\star)$ with $q^\star$ the closest Johnson--Viola--joinable weight~\cite{JV2013}. The source-certified advantage is
\begin{equation}
  \bigl[A(p)-V_Q(p)\bigr]_+
  =\tfrac14\Bigl[(3p-1)-\sqrt{(3p+1)(1-p)}\,\Bigr]\quad(p>\tfrac23),\qquad 0\ (p\le\tfrac23),
\end{equation}
i.e.\ $3\,C_1^{\mathrm{tomo}}(p)$ (the one-colluder tomographic capacity defined in Theorem~\ref{thm:exact-tomo}). The advantage onsets exactly at the extendibility
threshold $p^*=\tfrac23$; for $p<\tfrac23$ one has $V_Q(p)>A(p)$ (the colluder
over-coordinates, the source being $2$-shareable). By Proposition~\ref{prop:c1-local}
the broadcast/device-independent advantage is $0$ for all $p$.
\end{theorem}

\begin{proof}
$A(p)=\mathcal S(\rho_W(p))=\tfrac12(1-(-p))=\tfrac{1+p}{2}$. The colluder maximises
$\mathcal S(\rho_{AC})$ over $\rho_{AC}$ achievable with $\rho_{AB}=\rho_W(p)$ fixed; by
the octahedral symmetry of $G_{\mathrm{aligned}}$ and of the constraint the optimiser may
be taken Werner $\rho_W(q)$ (cf.\ Thm.~\ref{thm:exact-tomo}, Step~1), giving
$\mathcal S=\tfrac{1+q}{2}$, increasing in $q$. The largest joinable $q$ is $q^\star$ with
$p-q^\star=\tfrac12[(3p-1)-\sqrt{(3p+1)(1-p)}]$ (Johnson--Viola Cor.~III.3, as in
Thm.~\ref{thm:exact-tomo}), whence $A-V_Q=\tfrac{p-q^\star}{2}$. For $p\le\tfrac23$ the
symmetric point $q=p$ is joinable, so $V_Q\ge A$ and the advantage vanishes. A direct SDP
over all tripartite extensions confirms $V_Q(p)$ to four decimals and the sign change at
$p^*=\tfrac23$.
\end{proof}

\begin{corollary}[Composition over rounds]
\label{cor:c1-compose}
Over $n$ independent instances the advantage is maintained against collective
(permutation-symmetric) colluders with soundness error $\tilde O(D^2\log n/(C^2 n))$
(the multi-round extension detection theorem (Theorem~\ref{thm:b2a})) and amplifies exponentially against memoryless colluders
(the explicit finite-round detection protocol (Proposition~\ref{prop:exp-detection})).
\end{corollary}

\begin{remark}[Scope]
\label{rem:c1-scope}
Theorem~\ref{thm:c1} is a collusion-resistant \emph{coordination} (anti-impersonation)
guarantee, not a cryptographic key or randomness primitive: the latter require min-entropy
bounds and a security rate valid against coherent, memory-bearing colluders (via entropy accumulation), which we
do not claim. The non-definitional content is Proposition~\ref{prop:c1-local}: the
advantage is attributable solely to the source assumption, and the device-independent
side is constructively shown to be empty on identical local statistics.
\end{remark}

\subsection{Independence from Bell nonlocality}\label{sec:reversal}

%
%
%

We measure Bell nonlocality by the maximal CHSH value~
\cite{Horodecki2009}: for a two-qubit
state $\rho$ with correlation matrix $T_{ij}=\Tr(\rho\,\sigma_i\!\otimes\!\sigma_j)$,
$\mathrm{CHSH}_{\max}(\rho)=2\sqrt{t_1^2+t_2^2}$ with $t_1\ge t_2$ the two
largest singular values of $T$; $\rho$ is CHSH-Bell-nonlocal iff
$\mathrm{CHSH}_{\max}(\rho)>2$. We write $D_1$ for the one-colluder
shareability resource (Definition~\ref{def:Dm}) and $E_1$ for the set of
symmetrically $1$-extendible states.

\begin{theorem}[Shareability and Bell nonlocality order states differently]
\label{thm:reversal}
There exist two-qubit states $\rho,\rho'$ with $\rho$ CHSH-Bell-nonlocal,
$\rho'$ CHSH-Bell-local, and $D_1(\rho)<D_1(\rho')$. Hence $D_1$ is not a
monotonic function of $\mathrm{CHSH}_{\max}$ (nor of any quantity monotone
in it): the anti-collusion resource and Bell nonlocality induce distinct,
non-coincident orderings on quantum states.
\end{theorem}

\begin{proof}
Let $\rho=\tfrac14\bigl(I+0.1\,\sigma_x\!\otimes\!\sigma_x
-0.1\,\sigma_y\!\otimes\!\sigma_y+\sigma_z\!\otimes\!\sigma_z\bigr)$, the
Bell-diagonal state with correlation vector $(0.1,-0.1,1.0)$, and let
$\rho'=\rho_W(0.706)$ be the Werner state with singlet weight $p=0.706$.
The correlation matrix of $\rho$ has singular values $\{1,0.1,0.1\}$, so
$\mathrm{CHSH}_{\max}(\rho)=2\sqrt{1.01}\approx2.010>2$ (Bell-nonlocal),
while $\mathrm{CHSH}_{\max}(\rho')=2\sqrt2\,(0.706)\approx1.997<2$
(Bell-local). It remains to certify $D_1(\rho)<D_1(\rho')$.

\emph{Upper bound on $D_1(\rho)$.} Any $\sigma\in E_1$ satisfies
$D_1(\rho)\le\tfrac12\|\rho-\sigma\|_1$. Taking the explicit extendible
state obtained by symmetrising a feasible $2$-copy extension and projecting
it onto the positive cone yields
\begin{equation}
  D_1(\rho)\ \le\ U=0.00239 .
\end{equation}

\emph{Lower bound on $D_1(\rho')$.} For every Hermitian $M$ with
$-I\preceq M\preceq I$,
\begin{equation}
  D_1(\rho')=\min_{\sigma\in E_1}\tfrac12\|\rho'-\sigma\|_1
  \ \ge\ \tfrac12\Tr(M\rho')-\tfrac12\max_{\sigma\in E_1}\Tr(M\sigma).
  \label{eq:dual-lb}
\end{equation}
The collusive maximisation has the closed form
\begin{equation}
  \max_{\sigma\in E_1}\Tr(M\sigma)
   \;=\;\lambda_{\max}\!\bigl(\operatorname{sym}(M\otimes I_{B_1})\bigr),
  \qquad
  \operatorname{sym}(K)=\tfrac12\bigl(K+(I_A\!\otimes\!\mathrm{SWAP})\,K\,(I_A\!\otimes\!\mathrm{SWAP})\bigr),
  \label{eq:collusive-eig}
\end{equation}
because $\sigma\in E_1$ iff $\sigma=\Tr_{B_1}\Sigma$ for some swap-symmetric
$\Sigma\succeq0$, $\Tr\Sigma=1$, whence
$\Tr(M\sigma)=\Tr\!\bigl(\operatorname{sym}(M\otimes I)\,\Sigma\bigr)$, and
the maximum of a swap-invariant operator over swap-symmetric states is its
top eigenvalue (its top eigenvector lies in a $\pm1$ eigenspace of the
swap, whose rank-one projector is swap-symmetric). Choosing
$M=\operatorname{sign}(\rho'-\sigma^\star)$ for the (near-)optimal
$\sigma^\star$ and evaluating \eqref{eq:dual-lb}--\eqref{eq:collusive-eig}
gives
\begin{equation}
  D_1(\rho')\ \ge\ L=0.02950 ,
\end{equation}
which coincides with the exact Werner value
$\tfrac34(0.706-\tfrac23)=0.0295$.

Therefore $D_1(\rho)\le U=0.00239<L=0.02950\le D_1(\rho')$: the
Bell-nonlocal state $\rho$ is strictly less collusion-resistant than the
Bell-local state $\rho'$.
\end{proof}

\begin{remark}[The certificate is solver-independent]
\label{rem:reversal-cert}
Both bounds are eigenvalue computations on explicit matrices and do not
rely on the optimum reported by any SDP solver: $U$ uses an exactly
feasible extendible $\sigma$ (so $\tfrac12\|\rho-\sigma\|_1$ is an honest
upper bound), and $L$ uses an explicit witness $M$ with $\|M\|_\infty\le1$,
for which \eqref{eq:dual-lb}--\eqref{eq:collusive-eig} hold
unconditionally. The certified gap $L-U=0.02711$ exceeds the
double-precision error by ten orders of magnitude.
\end{remark}

\begin{remark}[A one-parameter family of witnesses, and the mechanism]
\label{rem:reversal-family}
The witness is not isolated. The ``classically anchored'' family $\rho_s$
with correlation vector $(s,-s,1)$, anchored at the classically-correlated
separable point $(0,0,1)$, is Bell-nonlocal for every $s>0$
($\mathrm{CHSH}_{\max}=2\sqrt{1+s^2}$) yet carries vanishing resource as
$s\to0^+$ (e.g.\ $D_1\approx0.0093$ at $s=0.2$, where
$\mathrm{CHSH}_{\max}\approx2.040$). Each such $\rho_s$ with small $s$ is
strictly less collusion-resistant than a Werner state in the
source-non-shareable window $p\in(\tfrac23,\tfrac1{\sqrt2})$, although the
Werner state is Bell-local; in the $(\mathrm{CHSH}_{\max},D_1)$ plane the
two families are not a single monotone curve. The mechanism is transparent:
Bell nonlocality is fixed by the correlation spectrum alone, whereas $D_1$
is fixed by the full extendibility geometry, and the anchored states sit
near the freely shareable point $(0,0,1)$, so their distance to $E_1$ stays
small however their spectrum tips them across the Bell bound.
\end{remark}

\begin{figure}[t]
\centering
\includegraphics[width=0.75\linewidth]{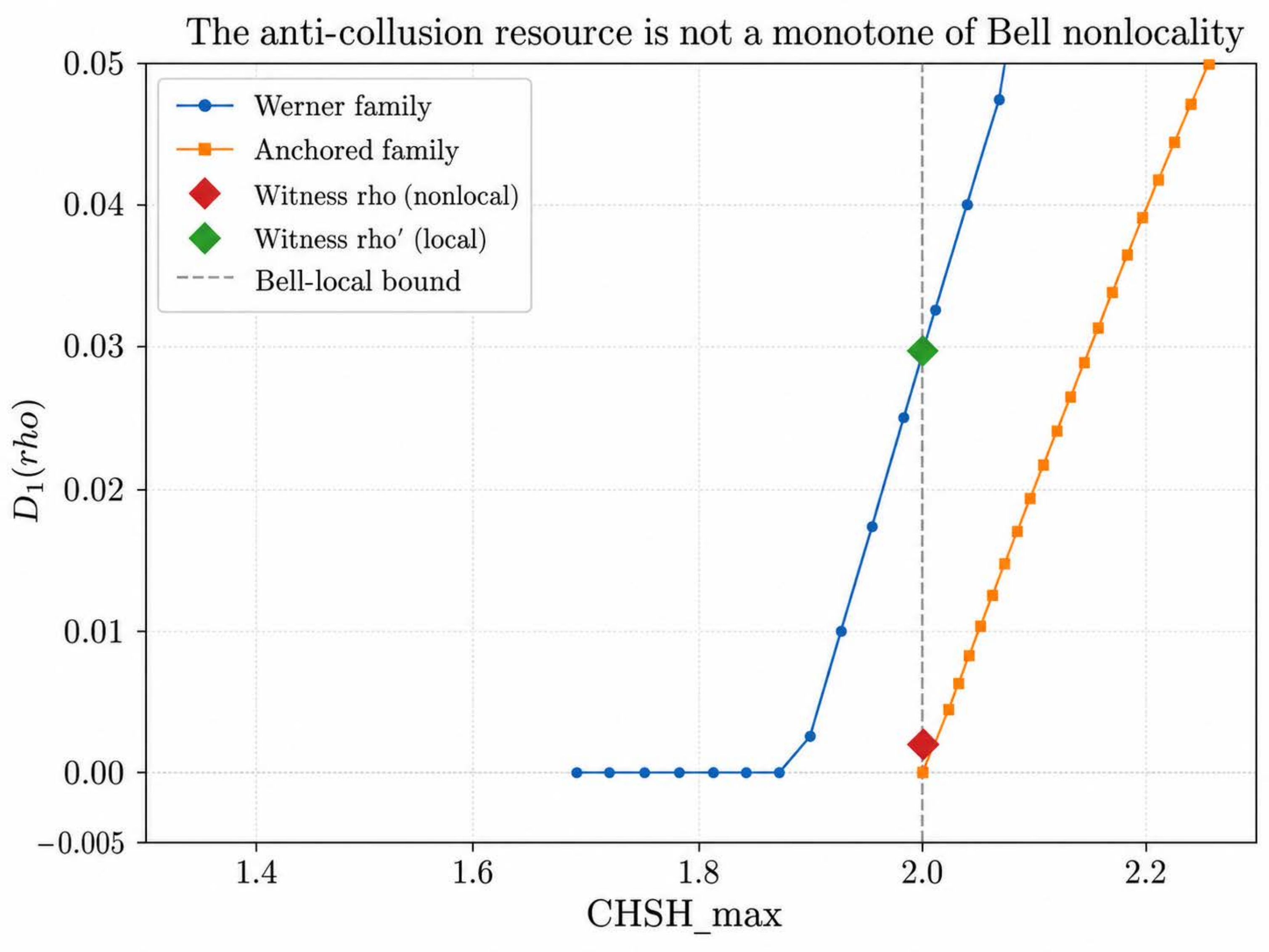}
\caption{\textbf{The anti-collusion resource is not a monotone of Bell nonlocality.} Two families in the $(\mathrm{CHSH}_{\max},D_1)$ plane: the Werner family (blue circles) and the ``classically anchored'' family $(s,-s,1)$ (orange squares). The two curves are not monotone with respect to one another. The witness pair of Theorem~\ref{thm:reversal} is marked: $\rho$ (red diamond, Bell-nonlocal with $\mathrm{CHSH}_{\max}\approx2.010$) has resource $D_1\approx0.0024$, while $\rho'$ (green diamond, Bell-local Werner with $p=0.706$) has resource $D_1\approx0.0295$. The certified gap $D_1(\rho')-D_1(\rho)\approx0.0271$ exceeds numerical error by ten orders of magnitude.}
\label{fig:reversal}
\end{figure}

\subsection{Resource theory of strategic non-shareability}\label{sec:resource-theory}

The state-level resource theory of unextendibility has been established by Kaur--Das--Wilde--Winter~\cite{KDWW2019,KDWW2021}, with quantifiers based on generalized divergences and robustness measures. Wang--Wilde~\cite{WangWilde2024} and Singh--Wilde~\cite{SinghWilde2025} subsequently developed unextendible entanglement monotones and channel bounds. These works take symmetric extendibility of the quantum state as the free set and define faithful monotones on the state space. Our contribution is to lift this framework to the \emph{behavior} or strategic layer: the free objects are still symmetrically extendible sources, but the resource is now the anti-collusion capacity---the optimal distinguishing advantage in a game---and the operational primitive is coordination with certified collusion resistance. The trace-distance resource $D_m$ on states induces the behavior-level capacity $C_{\rm AC}$, which is what an experimenter can certify from observed correlations.

We organize strategic shareability as a resource theory whose free objects are the freely shareable (symmetrically extendible) states and whose resource measure is the trace-distance monotone $D_m$.

%

\begin{definition}[Free set: symmetrically $m$-extendible states]
\label{def:free-set}
Fix finite local dimensions and let $m\ge 1$ be the number of colluders.
A bipartite state $\rho_{AB}$ is \emph{symmetrically $m$-extendible on
$B$} if there exists a state $\rho_{AB_0B_1\cdots B_m}$ on
$\mathcal H_A\otimes\mathcal H_B^{\otimes (m+1)}$ that is
(i) invariant under every permutation of $B_0,B_1,\dots,B_m$, and
(ii) satisfies $\Tr_{B_1\cdots B_m}\rho_{AB_0B_1\cdots B_m}=\rho_{AB}$.
We denote the set of such states by $\mathcal E_m$. The sets are nested,
$\mathcal E_1\supseteq\mathcal E_2\supseteq\cdots$, each $\mathcal E_m$ is
convex and compact, and every separable state lies in $\bigcap_m\mathcal E_m$.
\end{definition}

\begin{definition}[Trace-distance resource]
\label{def:Dm}
The \emph{$m$-shareability resource} of $\rho_{AB}$ is
\begin{equation}
  D_m(\rho_{AB})\;:=\;\min_{\sigma\in\mathcal E_m}\tfrac12\,\bigl\|\rho_{AB}-\sigma\bigr\|_1 ,
\end{equation}
the minimum being attained because $\mathcal E_m$ is compact.
\end{definition}

\begin{definition}[Free operations]
\label{def:free-ops}
A CPTP map $\Phi$ acting on $AB$ is \emph{$m$-shareability-free} if it maps
the free set into itself, $\Phi(\mathcal E_m)\subseteq\mathcal E_m$. We write
$\mathcal O_m$ for this class. This is the maximal class of operations that
cannot create the resource; Lemma~\ref{lem:free-ops-concrete} exhibits a
physically motivated subclass, and restricting to any subclass only
strengthens the monotonicity statement below.
\end{definition}

\begin{lemma}[Concrete free operations]
\label{lem:free-ops-concrete}
The class $\mathcal O_m$ contains:
\begin{enumerate}
\item[(a)] local channels $\Phi_A\otimes\Phi_B$ for arbitrary CPTP maps
      $\Phi_A,\Phi_B$;
\item[(b)] mixing with a free state,
      $\rho\mapsto(1-q)\rho+q\,\sigma_0$ with $\sigma_0\in\mathcal E_m$,
      $q\in[0,1]$;
\item[(c)] tensoring with a free ancilla and discarding subsystems
      compatibly with the $A{:}B$ cut;
\item[(d)] every entanglement-breaking channel across $A{:}B$.
\end{enumerate}
\end{lemma}

The proof is immediate from the defining property $\Phi(\mathcal E_m)\subseteq\mathcal E_m$ and contractivity of trace distance; details are given in Supplementary Note~13.

\begin{theorem}[$D_m$ is a resource monotone]
\label{thm:monotone}
For every $\Phi\in\mathcal O_m$ and every state $\rho$,
\begin{equation}
  D_m\!\bigl(\Phi(\rho)\bigr)\;\le\;D_m(\rho).
\end{equation}
In particular $D_m$ does not increase under any of the operations in
Lemma~\ref{lem:free-ops-concrete}.
\end{theorem}

The proof uses contractivity of trace distance under CPTP maps; see Supplementary Note~13.

\begin{proposition}[Basic axioms of the resource measure]
\label{prop:axioms}
The functional $D_m$ satisfies:
\begin{enumerate}
\item[(i)] \emph{Faithfulness:} $D_m(\rho)=0$ if and only if
      $\rho\in\mathcal E_m$;
\item[(ii)] \emph{Convexity:}
      $D_m\!\bigl(\sum_i p_i\rho_i\bigr)\le\sum_i p_i\,D_m(\rho_i)$ for any
      probability distribution $\{p_i\}$;
\item[(iii)] \emph{Closed form on the Werner line:}
      $D_m\!\bigl(\rho_W(p)\bigr)=\tfrac34\,(p-p_m^\star)_+$ with
      $p_m^\star=(m+3)/[3(m+1)]$.
\end{enumerate}
\end{proposition}

Faithfulness and convexity follow from compactness and convexity of $\mathcal E_m$ together with the metric property of the trace norm. The closed form (iii) is proved in Supplementary Note~15. The full resource-theoretic proofs are given in Supplementary Note~13.

\begin{remark}[Behavior-level monotone]
\label{rem:behavior-level}
The same argument applies verbatim at the behavior level. The
anti-collusion capacity
$C^{\mathrm R}_{\mathrm{AC}}(P_{12})
=\inf_{Q\in\mathcal S_{\mathrm R}(P_{12})}d_{\mathrm{TV}}(P_{12},Q)$
is monotone under the free behavior operations consisting of local
classical pre- and post-processing of the players' inputs and outputs and
mixing with elements of the collusive shadow: these operations map
$\mathcal S_{\mathrm R}$ into itself, and $d_{\mathrm{TV}}$ is contractive
under stochastic maps. Thus both the state-level resource $D_m$ and the
behavior-level capacity $C^{\mathrm R}_{\mathrm{AC}}$ are genuine monotones
of the resource theory of strategic shareability.
\end{remark}

%
%
%
%

We now lift the exact trace-distance resource to arbitrary local
dimension. The extendibility thresholds are the Johnson--Viola results~[46];
the closed-form constants and the resulting resource monotones are new.

\paragraph{Families.}
For local dimension $d$ write $V=\sum_{ij}\ket{ij}\!\bra{ji}$ for the swap,
$P_-=(I-V)/2$ and $P_+=(I+V)/2$ for the antisymmetric and symmetric
projectors, with $\operatorname{rank}P_-=\binom{d}{2}$ and
$\operatorname{rank}P_+=\binom{d+1}{2}$. Define the $d$-dimensional Werner
and isotropic families
\begin{align}
  \rho_W^{(d)}(p) &= p\,\frac{P_-}{\binom{d}{2}} + (1-p)\,\frac{I}{d^2},
   & \rho_{\mathrm{iso}}^{(d)}(p) &= p\,\ket{\Phi_d^+}\!\bra{\Phi_d^+} + (1-p)\,\frac{I}{d^2},
\end{align}
with $p\in[0,1]$ and $\ket{\Phi_d^+}=\tfrac1{\sqrt d}\sum_i\ket{ii}$. Both
reduce at $d=2$ to the qubit Werner state of the main text (singlet weight
$p$ mixed with $I/4$). In terms of the Johnson--Viola parameters,
$\Psi^-(p)=\Tr[V\rho_W^{(d)}(p)]=\tfrac1d-p\,\tfrac{d+1}{d}$ and
$\Phi^+(p)=\tfrac1d\bigl(p(d^2-1)+1\bigr)$, so $p\in[0,1]$ covers
$\Psi^-\in[-1,\tfrac1d]$ and $\Phi^+\in[\tfrac1d,d]$.

\begin{lemma}[Trace norm between two Werner / isotropic states]
\label{lem:Wernerdist-d}
For all $p,q\in[0,1]$,
\begin{align}
  \tfrac12\bigl\|\rho_W^{(d)}(p)-\rho_W^{(d)}(q)\bigr\|_1
    &= c_W(d)\,|p-q|, & c_W(d) &= \frac{d+1}{2d},\\[2pt]
  \tfrac12\bigl\|\rho_{\mathrm{iso}}^{(d)}(p)-\rho_{\mathrm{iso}}^{(d)}(q)\bigr\|_1
    &= c_{\mathrm{iso}}(d)\,|p-q|, & c_{\mathrm{iso}}(d) &= \frac{d^2-1}{d^2}.
\end{align}
In particular $c_W(2)=c_{\mathrm{iso}}(2)=\tfrac34$.
\end{lemma}

The spectral calculation is given in Supplementary Note~13.

\begin{theorem}[Exact resource in arbitrary dimension]
\label{thm:Dm-d}
For $m$ colluders there are $m+1$ $B$-side copies (one authorized plus $m$ colluders). Translating
the Johnson--Viola $1$-$(m+1)$ sharability thresholds~[46] into the present
parametrization, the $m$-colluder extendibility thresholds are
\begin{align}
  p^{*}_m(d) &= \frac{d(d-1)+m+1}{(m+1)(d+1)}
   = \frac{d^2-d+m+1}{(m+1)(d+1)} && (\text{Werner}),\\[2pt]
  p^{*}_{m,\mathrm{iso}}(d) &= \frac{m+1+d}{(m+1)(d+1)} && (\text{isotropic}),
\end{align}
and the $m$-shareability resource is exactly
\begin{align}
  D_m\!\bigl(\rho_W^{(d)}(p)\bigr)
   &= \frac{d+1}{2d}\,\bigl(p-p^{*}_m(d)\bigr)_+,\\[2pt]
  D_m\!\bigl(\rho_{\mathrm{iso}}^{(d)}(p)\bigr)
   &= \frac{d^2-1}{d^2}\,\bigl(p-p^{*}_{m,\mathrm{iso}}(d)\bigr)_+.
\end{align}
At $d=2$ both thresholds equal $(m+3)/[3(m+1)]$ and
$D_m=\tfrac34(p-p^{*}_m)_+$, recovering the main-text formula.
\end{theorem}

The proof uses Werner/isotropic twirling and the Johnson--Viola $1$-$(m+1)$ sharability thresholds; see Supplementary Note~13.

\begin{remark}[Dimensional onset of the resource]
\label{rem:dim-onset}
The one-colluder resource is special to qubits. Since
$p^{*}_m(d)<1\iff m\ge d-1$, the $m$-colluder Werner resource
$D_m(\rho_W^{(d)}(\cdot))$ is identically zero on the whole family unless
$m\ge d-1$. Equivalently, even the most-entangled qudit Werner state
($p=1$, $\Psi^-=-1$) is $1$-$(d-1)$ sharable -- it is the two-party
reduction of the totally antisymmetric $d$-party state, which furnishes the
sharing state -- so a positive resource first appears only at $n=d$ copies,
i.e.\ $m=d-1$ colluders. For $d=2$ this gives the familiar onset at one
colluder; for $d=3$ one needs $m\ge2$, etc.
\end{remark}

\begin{remark}[Dimensional behaviour of the constants and limits]
\label{rem:dim-limits}
The trace-distance constants coincide at $d=2$ ($c_W=c_{\mathrm{iso}}=3/4$,
consistent with the local equivalence of qubit Werner and isotropic states)
but separate for $d>2$: $c_W(d)=\tfrac{d+1}{2d}\downarrow\tfrac12$ while
$c_{\mathrm{iso}}(d)=1-d^{-2}\uparrow1$ as $d\to\infty$. Both thresholds
tend to the separability point $1/(d+1)$ as $m\to\infty$ (arbitrary
sharability equals separability), consistent with the Werner/isotropic
separability boundaries $\Psi^-\ge0$ and $\Phi^+\le1$.
\end{remark}

\subsection{Exactly solvable calibration: Werner and isotropic capacities}\label{sec:werner}

The framework is exactly solvable on Werner and isotropic source families, turning the Johnson--Viola shareability thresholds into closed-form resource monotones. These formulas serve as a calibration of the theory rather than a new extendibility result. The certification hierarchy is: \emph{score-only} (weakest, Supplementary Note~6), \emph{source-level} (Theorem~\ref{thm:Dm-d} and Corollary~\ref{thm:werner-threshold}), and \emph{tomographic} (strongest, Theorem~\ref{thm:exact-tomo}). More information yields an earlier certified onset; the gap between the CHSH-score onset ($S_{12}=2$) and the source-level onset ($p_m^{\ast}$) is the Bell-local but source-non-shareable window. We organize the results into a certification hierarchy: \emph{score-only} certification (Supplementary Note~6) uses only the observed CHSH score and is the weakest certificate; \emph{source-level} certification (Theorem~\ref{thm:Dm-d} and Corollary~\ref{thm:werner-threshold}) uses knowledge of the Werner family; \emph{tomographic} certification (Theorem~\ref{thm:exact-tomo}) uses the full behavior statistics under a tomographically complete measurement set. More information yields an earlier certified onset. The gap between the CHSH-score onset ($S_{12}=2$) and the source-level onset ($p_m^{\ast}$) is precisely the Bell-local but source-non-shareable window. We show that the relevant boundary for the stronger certificates is symmetric extendibility of the source, not Bell nonlocality.

\paragraph{Werner state.} Consider the two-qubit Werner state
\[
\rho_W(p)=pP_{-}+(1-p)\frac{I}{4},\qquad
P_{-}=\frac{I-F}{2},
\]
where $F$ is the swap operator on two qubits and $p\in[0,1]$. Let $m$ be the number of colluders.

The symmetric-extendibility threshold is a known shareability result~\cite{JV2013,JSZ2022}. We include it here because it is the point at which the anti-collusion resource turns on in the quantum-source model, and because its two-qubit form leads directly to an exact closed-form resource monotone.

\begin{corollary}[Multi-colluder symmetric-extendibility threshold]\label{thm:werner-threshold}
The Werner state $\rho_W(p)$ is symmetrically $m$-extendible (Definition~\ref{def:free-set}) if and only if
\[
p\le p_m^{\ast}=\frac{m+3}{3(m+1)}.
\]
In particular, $p_1^{\ast}=2/3$, $p_2^{\ast}=5/9$, and $\lim_{m\to\infty}p_m^{\ast}=1/3$.
\end{corollary}

Thus Corollary~\ref{thm:werner-threshold} should not be read as a new theorem about Werner extendibility. Its role here is operational: it identifies the exact source-level anti-collusion onset once the colluder is restricted to a genuine quantum extension of the mediator's source. The one-sided $1$-to-$(m+1)$ sharability problem was solved for arbitrary local dimension by Johnson and Viola~\cite{JV2013}, and the spin/swap-Hamiltonian route we use here is the same one later developed by Jakab, Solymos and Zimbor\'as~\cite{JSZ2022}. We reprove the two-qubit case in Supplementary Note~14 for completeness. Our contribution is not the threshold itself but its operational reinterpretation: we show that this extendibility boundary is exactly the onset of anti-collusion resource, quantify it by the game-weighted trace-distance monotone $D_m(p)$ (Theorem~\ref{thm:Dm-d} at $d=2$), and lift it to a certifiable behavior-level statement under tomographically complete measurements (Theorem~\ref{thm:exact-tomo}).

The threshold $p_m^{\ast}$ lies strictly below the CHSH nonlocality threshold $p_{\rm Bell}=1/\sqrt{2}\approx0.707$ for all finite $m$. Hence, whenever
\[
\frac{m+3}{3(m+1)}<p\le\frac{1}{\sqrt{2}},
\]
the Werner source is Bell-local with respect to CHSH, but already non-shareable under quantum-source extensions. This statement is source-certified: the same observed Bell-local behavior would be freely extendible in a broadcast-enabled hidden-variable model unless the source is specified or tomographically certified.

\begin{figure}[t]
\centering
\includegraphics[width=\linewidth]{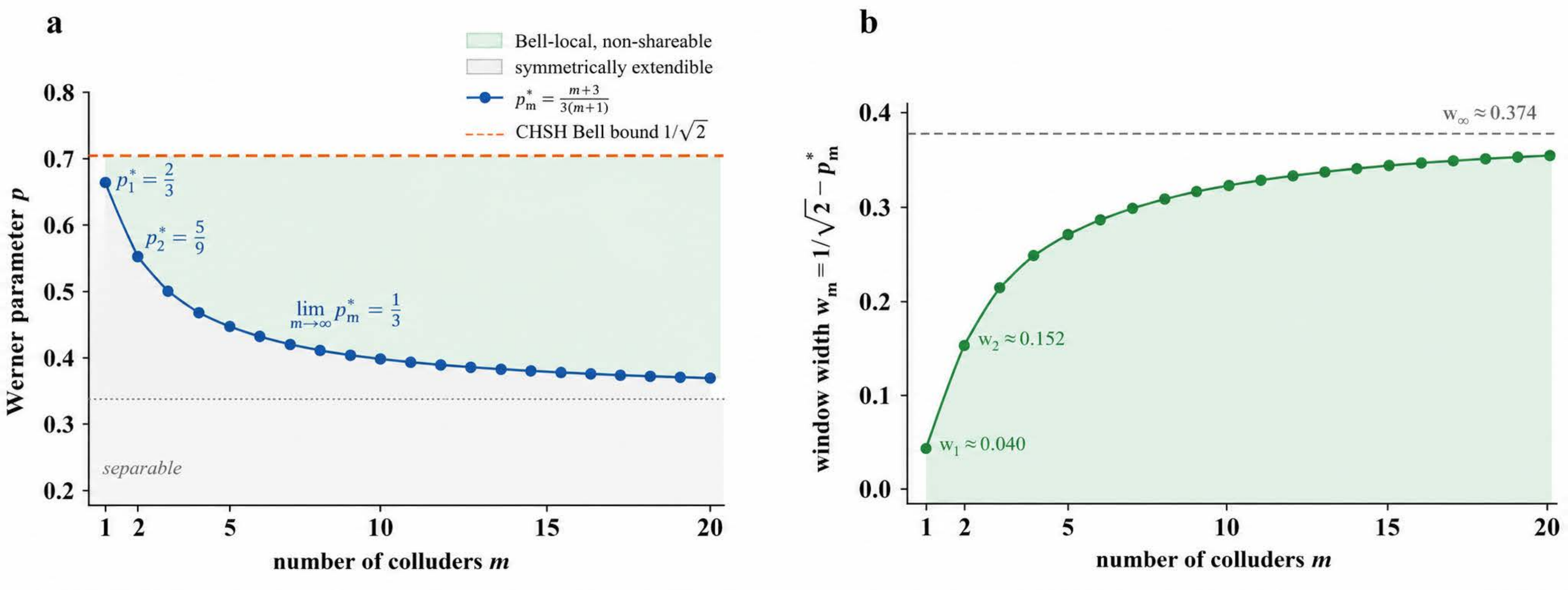}
\caption{\textbf{Symmetric-extendibility threshold and Bell-local resource window for Werner sources.} \textbf{a}, The multi-colluder symmetric-extendibility threshold $p_m^{\ast}=\frac{m+3}{3(m+1)}$ for qubits (blue) decreases monotonically with the number of colluders $m$, while the CHSH Bell-local bound $1/\sqrt{2}$ (orange, dashed) remains constant. The green shaded region between them is the Bell-local but source-non-shareable window. Theorem~\ref{thm:Dm-d} generalizes the exact threshold and trace-distance resource to arbitrary local dimension $d$. The gray region below $p_m^{\ast}$ is symmetrically extendible; the light-gray region below $1/3$ is separable. \textbf{b}, Width $w_m=1/\sqrt{2}-p_m^{\ast}$ of the Bell-local non-shareable window as a function of $m$. For one colluder $w_1\approx0.040$; for two colluders $w_2\approx0.152$; in the limit $m\to\infty$ the width approaches $w_{\infty}\approx0.374$, covering almost the entire entangled-but-Bell-local range.}
\label{fig:werner-thresholds}
\end{figure}

At $d=2$, Theorem~\ref{thm:Dm-d} reduces to the qubit formula
\[
D_m(p)=\frac{3}{4}\Bigl(p-\frac{m+3}{3(m+1)}\Bigr)_+,
\]
where $(x)_+=\max\{x,0\}$. The proof uses Werner twirling and the eigenstructure of $P_{-}-I/4$ (Supplementary Note~15). The resource rises linearly from zero at $p_m^{\ast}$ and passes through the Bell point $p=1/\sqrt{2}$ with no kink, confirming that the resource onset is governed by extendibility, not by the Bell boundary.

\paragraph{Bell-local but source-non-shareable window.} The maximal CHSH value for Werner states is $S_{\max}(p)=2\sqrt{2}\,p$, so the CHSH Bell-nonlocality threshold is $p>1/\sqrt{2}$. By Corollary~\ref{thm:werner-threshold}, source non-shareability begins at $p>p_m^{\ast}$. The resulting Bell-local but source-non-shareable interval is
\[
\boxed{\;\frac{m+3}{3(m+1)}<p\le\frac{1}{\sqrt{2}}.\;}
\]
For one colluder this is $(2/3,1/\sqrt{2}]$, width $\approx0.040$; for two colluders it is $(5/9,1/\sqrt{2}]$, width $\approx0.152$; in the many-colluder limit it approaches $(1/3,1/\sqrt{2}]$, covering almost the entire entangled-but-Bell-local Werner range.

\begin{figure}[t]
\centering
\includegraphics[width=\linewidth]{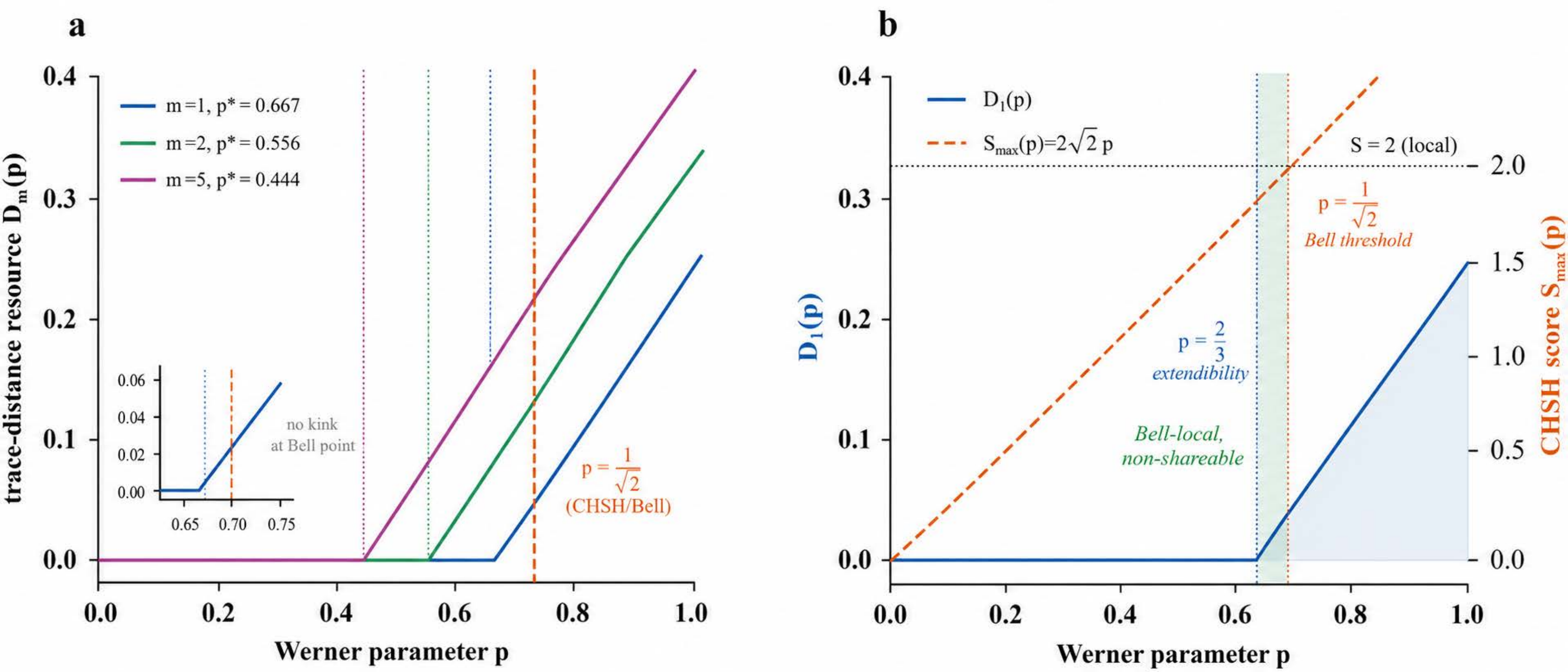}
\caption{\textbf{Exact trace-distance resource on the Werner line.} \textbf{a}, The qubit trace-distance resource $D_m(p)=\frac{3}{4}(p-p_m^{\ast})_+$ for $m=1$ (blue), $m=2$ (green) and $m=5$ (purple). Theorem~\ref{thm:Dm-d} gives the exact resource $c_W(d)(p-p_m^{\ast})_+$ with $c_W(d)=(d+1)/2d$ for arbitrary local dimension, recovering $3/4$ at $d=2$. Each curve is a straight line rising from zero at its corresponding extendibility threshold $p_m^{\ast}$. The vertical orange dashed line marks the CHSH Bell threshold $p=1/\sqrt{2}$; all curves pass through it smoothly with no kink, confirming that the resource onset is governed by symmetric extendibility, not by the Bell boundary. The inset zooms into the Bell-point neighbourhood. \textbf{b}, Direct comparison of the $m=1$ resource $D_1(p)$ (blue, left axis) and the maximal CHSH score $S_{\max}(p)=2\sqrt{2}\,p$ (orange, dashed, right axis). The two onsets occur at different thresholds ($p=2/3$ versus $p=1/\sqrt{2}$), and the green shaded interval $(2/3,1/\sqrt{2}]$ is the Bell-local but source-non-shareable window. Theorem~\ref{thm:exact-tomo} lifts this to the exact tomographic capacity $C_1^{\mathrm{tomo}}(p)=(1/12)[(3p-1)-\sqrt{(3p+1)(1-p)}]$, removing the previous conservative constant.}
\label{fig:trace-distance}
\end{figure}

\begin{theorem}[Exact tomographic capacity, one colluder]\label{thm:exact-tomo}
For $p\in[\tfrac23,1]$,
\[
C^{\mathrm{tomo}}_1(p)
\;=\;\frac{1}{12}\Bigl[(3p-1)-\sqrt{(3p+1)(1-p)}\,\Bigr],
\]
and $C^{\mathrm{tomo}}_1(p)=0$ for $p\le\tfrac23$. The onset is at
$p^{*}_1=\tfrac23$, $C^{\mathrm{tomo}}_1(1)=\tfrac16$, and the curve rises
with vertical tangent at $p=1$. This exceeds the conservative bound
$\tfrac1{6\sqrt{15}}(p-\tfrac23)$ by a factor $\approx 8$ near onset,
growing to $\approx 11.6$ at $p=1$.
\end{theorem}

In light of the state-level resource theory, $C_1^{\mathrm{tomo}}$ may be viewed as the behavior-level projection, under the fixed tomographic measurement set, of the state-level unextendible-entanglement measures of~\cite{WangWilde2024}.

\begin{proof}
See Supplementary Note~17.
\end{proof}

\begin{remark}[Numerical certification and general $m$]
\label{rem:tomo-sdp}
A direct semidefinite program minimising $d_{\mathrm{TV}}$ over all
tripartite extensions --- making no Werner assumption --- reproduces the
closed form to five decimals throughout $p\in[\tfrac23,1]$. For $m$ colluders the capacity is the analogous SDP; its resource onset occurs exactly at the extendibility
threshold $p^{*}_m=(m+3)/[3(m+1)]$ (verified, e.g.\ $p^{*}_2=5/9$). The
one-colluder case is the one admitting the closed form above.
\end{remark}

\begin{figure}[t]
\centering
\includegraphics[width=0.75\linewidth]{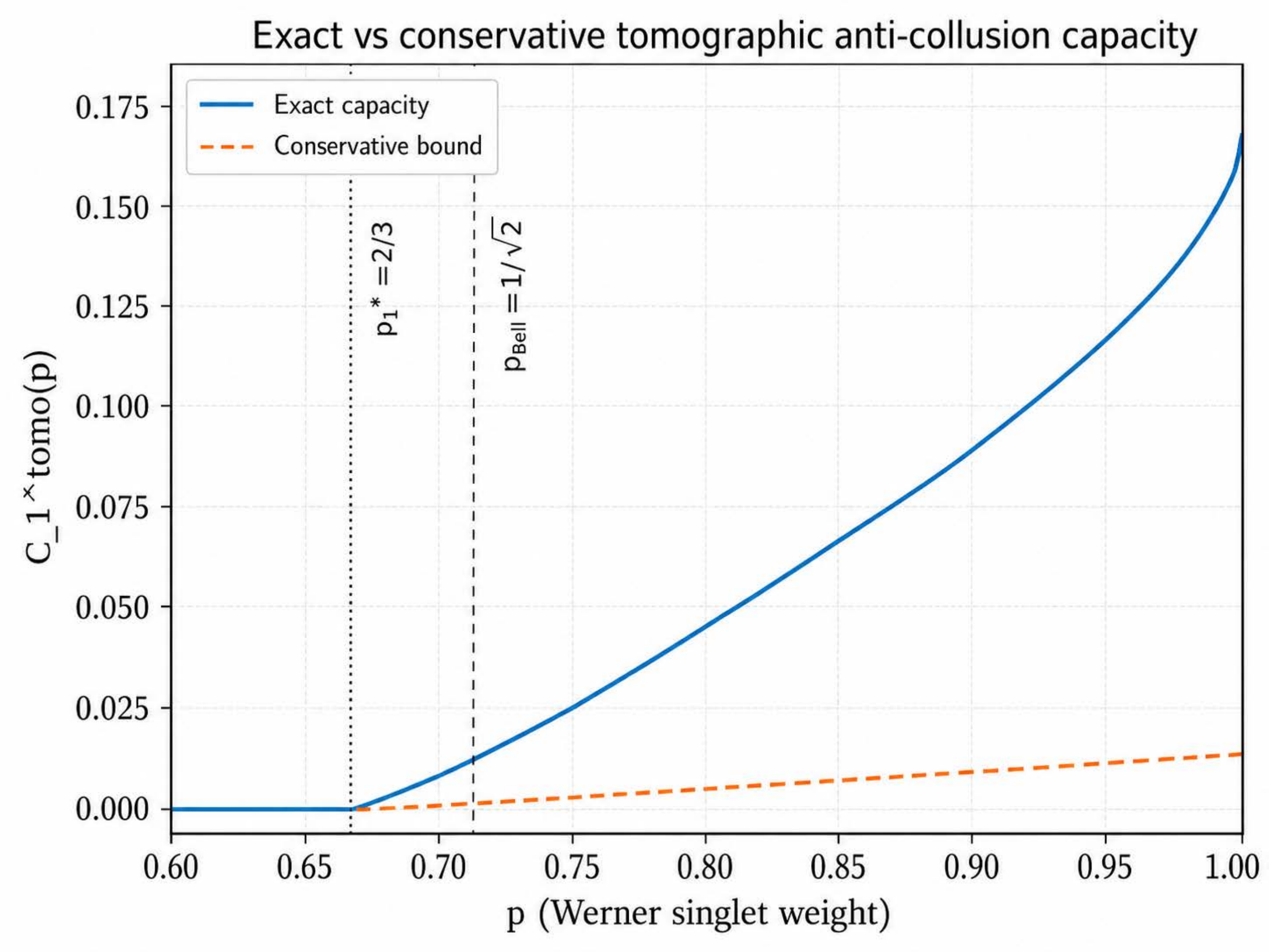}
\caption{\textbf{Exact versus conservative tomographic anti-collusion capacity.} The exact one-colluder capacity $C_1^{\mathrm{tomo}}(p)=(1/12)[(3p-1)-\sqrt{(3p+1)(1-p)}]$ (blue) replaces the previous intentionally conservative lower bound $(p-2/3)/(6\sqrt{15})$ (orange dashed). The exact curve exceeds the conservative bound by a factor $\approx8$ near onset, growing to $\approx11.6$ at $p=1$. The onset is exactly at the extendibility threshold $p^*_1=2/3$ (dotted); the CHSH Bell threshold $1/\sqrt{2}$ (dashed) lies strictly above.}
\label{fig:exact-tomo}
\end{figure}

\subsection{Supporting diagnostics and score-only certification}\label{sec:diagnostics}

The preceding sections established a separation theorem, an ordering crossing, and exact calibrations. Here we collect the supporting technical results: score-only certification, multi-round composition, and finite-data diagnostics.

\paragraph{Score-only CHSH frontier.}\label{sec:frontier} If the verifier observes only a single CHSH score, the strongest universal certificate is constrained by Bell-correlation monogamy. In this weaker certification model the onset coincides with the Bell-local boundary, which is exactly why a finer (source-level or tomographic) certificate is needed. The Toner--Verstraete monogamy inequality $S_{12}^2+S_{13}^2\le 8$ gives the exact score-certified frontier $\Gamma_{\rm CHSH}^{+}(s)=[(s-\sqrt{8-s^2})/8]_+$, which is positive if and only if $s>2$. The full derivation, tightness construction, and game separation are given in Supplementary Notes~6--8.

\paragraph{Multi-round security.} The single-shot anti-collusion capacity governs multi-instance guarantees.

\begin{lemma}[Non-shareability persists under independent repetition]\label{lem:multi-round-floor}
For every $n\ge1$ and every extension class $\mathsf R$ closed under marginalization,
\[
\mathcal C_{\rm AC}^{\mathsf R}(P_{12}^{\otimes n})\ge\mathcal C_{\rm AC}^{\mathsf R}(P_{12}).
\]
In particular, if the single instance is non-shareable, then the $n$-instance behavior is non-shareable for every $n$, even against colluders using arbitrary cross-round correlations.
\end{lemma}

\begin{proposition}[Exponential detection, memoryless colluders]\label{prop:exp-detection}
Let the colluder use any product strategy $Q=\bigotimes_{k=1}^{n}Q^{(k)}$ with $Q^{(k)}\in\mathcal S_{\mathsf R}(P_{12})$. Then
\[
d_{\rm TV}(P_{12}^{\otimes n},Q)\ge1-(1-\mathcal C_{\rm AC}^{\mathsf R}(P_{12})^2)^{n/2}.
\]
Equivalently, the soundness error decays exponentially with rate set by the single-shot capacity.
\end{proposition}

\begin{theorem}[Vanishing soundness against collective colluders]\label{thm:b2a}
Let $C=\mathcal C_{\rm AC}^{\mathsf R}(P_{12})>0$ and let the colluder be permutation-symmetric across $n$ rounds, with per-round joint dimension $D$. Then
\[
d_{\rm TV}(P_{12}^{\otimes n},Q^{(n)})\ge 1-\beta_n,\qquad \beta_n=\tilde O\!\left(\frac{D^2\log n}{C^2 n}\right).
\]
The soundness error against the collective class therefore vanishes polynomially in $n$, with the per-round detection rate set by the single-shot capacity $C$.
\end{theorem}

A security rate valid against coherent, memory-bearing colluders (via entropy accumulation) is left to future work.

\paragraph{Finite-data and NPA diagnostics.} The monotone certification map $s\mapsto\Gamma^+_{\rm CHSH}(s)$ is compatible with any lower-confidence-bound estimator, so finite-data certification is immediate (see Methods). Beyond the analytically solvable CHSH and Werner cases, collusive vulnerability is an extendibility SDP; finite-level NPA gives diagnostic outer bounds (Supplementary Note~12).

\section{Discussion}\label{sec:discussion}

The picture that emerges is that strategic non-shareability is a resource in the technical sense: a free set (the symmetrically-extendible sources), a class of free operations that cannot create it, and a faithful, convex monotone $D_m$ that is non-increasing under those operations and exactly computable on the Werner and isotropic lines in any dimension. Its onset is dimension-gated---a positive one-colluder resource is special to qubits, and $d$-dimensional sources require $m\ge d-1$ colluders before the resource appears---reflecting that the maximally-entangled qudit Werner state is itself $(d-1)$-shareable.

The two orderings strictly cross (Theorem~\ref{thm:reversal}), so shareability and Bell nonlocality are genuinely independent resources. This is the upgrade from ``a nested window exists'' to ``the orderings cross'': the anti-collusion resource is not a monotone function of any Bell-nonlocality measure.

The single-shot capacity controls multi-round guarantees. For $n$ independent instances the joint capacity never falls below the single-shot value, even against colluders that correlate across rounds (Lemma~\ref{lem:multi-round-floor}). Against memoryless colluders the distinguishing advantage grows exponentially (Proposition~\ref{prop:exp-detection}); against collective (permutation-symmetric) colluders a de~Finetti reduction gives soundness error $\tilde O(D^2\log n/(C^2 n))$, vanishing in $n$ (Theorem~\ref{thm:b2a}). A security rate valid against coherent, memory-bearing colluders (via entropy accumulation) is left to future work.

For the observed (aligned) behaviour, device-independent certification provably yields nothing: the aligned Werner behaviour has an LHV model (Proposition~\ref{prop:c1-local}), hence is DI-extendible for every $p$. Yet a source-certified quantum guarantee exists, is operationally useful---the coordination game $G_{\mathrm{aligned}}$ gives a strictly positive advantage over every quantum-source colluder for $p>2/3$ that persists and amplifies over rounds, while the broadcast/DI advantage is identically zero (Theorem~\ref{thm:c1}). The source assumption therefore buys a concrete, certifiable functionality that DI cannot deliver here.

\paragraph{Limitations.} These are source-certified, not device-independent: below the Bell bound a hidden variable can be broadcast; the positive window requires a specified or tomographically reconstructed source. The tilted-CHSH NPA curves are finite-level diagnostic outer bounds, not certified frontiers. A fully memory-tolerant security rate valid against coherent, memory-bearing colluders (via entropy accumulation), and certified (not merely diagnostic) tilted-CHSH frontiers with matching strategies, remain open.

Future work should close the gap between NPA upper bounds and explicit lower-bound strategies for non-CHSH inequalities, develop fully memory-tolerant device-independent confidence analyses, and extend the exact capacity formula to arbitrary $m$ colluders. Experimentally, the relevant task is not simply to violate a Bell inequality, but to certify a positive shareability deficit from observed authorized correlations. Candidate platforms include photonic polarization qubits, trapped ions, superconducting qubits, or nitrogen-vacancy based networks.
\section{Methods}\label{sec:methods}

\subsection{Classical copied-seed extension}
Classical hidden-variable mediators are freely shareable because the hidden seed can be copied to any number of colluders without disturbing the authorized marginal. For every classical mediator in the copied-seed threat model, the colluder can reproduce the authorized score in any relabelled test, giving zero anti-collusion power.

\subsection{Distance theorem}
Supplementary Note~3 gives the full proof of Theorem~\ref{thm:distance}. The proof uses: finite-alphabet compactness of the behavior sets; convexity of the collusive shadow; Sion's minimax theorem; and the variational characterization of total-variation distance~\cite{Uhlmann1976,Fuchs1999}. The basic properties needed for the framework are: (i) freely shareable correlations have zero anti-collusion power; (ii) positive anti-collusion power witnesses task-relevant non-shareability; and (iii) collusive vulnerability is concave under unflagged mixing.

\subsection{CHSH monogamy frontier}
Let $\mathcal{B}_{12}$ and $\mathcal{B}_{13}$ be the two CHSH operators sharing player~1's observables. The standard Toner--Verstraete monogamy inequality $S_{12}^2+S_{13}^2\le 8$ bounds the squared expectation values. Supplementary Note~6 states the operator form, derives the tightness construction, and gives the exact certified anti-collusion power formula. Exact monogamy of a pure authorized marginal states that if $\rho_{12}=\Tr_3\rho_{123}$ is pure, then $\rho_{123}=\rho_{12}\otimes\rho_3$. Because the maximally entangled Bell state $\ket{\Phi_2}$ has maximally mixed local marginal $\Tr_B\dyad{\Phi_2}=I/2$, Alice's outcome distribution is uniform for any local observable. After the factorization $\rho_{123}=\dyad{\Phi_2}_{12}\otimes\rho_3$, the colluder's system is product with Alice's, so any collusive response rule gives $C_{13}=1/2$ regardless of the colluder's strategy.

\subsection{Finite-data estimator}
We estimate the CHSH score either by estimating four correlators separately or by using a single-trial unbiased estimator $Z_i=4(-1)^{xy}ab$. The main-text confidence interval uses the correlator-wise estimator with a union bound; the Supplement gives both variants. Supplementary Note~10 gives the derivation of the Hoeffding-based confidence bound, together with empirical-Bernstein alternatives and a proof of the monotonicity of $\Gamma_{\rm CHSH}^{+}(s)$ on $[0,2\sqrt2]$.

\subsection{NPA relaxation and diagnostics}

Semidefinite hierarchies such as the Doherty--Parrilo--Spedalieri hierarchy for symmetric extensions~\cite{DPS2002,DPS2004} and the NPA hierarchy for quantum correlations~\cite{NPA2007,NPA2008} provide natural relaxation tools for extension and marginal-compatibility problems. The tilted-CHSH collusive vulnerability is bounded by an NPA level-2 semidefinite relaxation implemented in CVXPY. For three parties with binary $\pm1$ observables $A_x,B_y,C_z$ and projective-measurement substitutions $A_x^2=B_y^2=C_z^2=I$, the level-2 moment matrix is $22\times22$ and indexed by words of length at most~2 drawn from $\{I,A_0,A_1,B_0,B_1,C_0,C_1\}$. Different-party operators commute. The tilted-CHSH family self-tests partially entangled two-qubit states at maximal violation~\cite{Acin2012,Bamps2015,Yang2013,Kaniewski2016}. The relaxed problem is
\[
V_{13}^{{\rm NPA},2}(s;\alpha)=\max_{\Gamma\succeq0}\;I_\alpha^{13}
\]
subject to $\Gamma\succeq0$, normalization $\Gamma_{\emptyset,\emptyset}=1$, algebraic substitutions, and the authorized-score constraint $I_\alpha^{12}\ge s$.

The SDP was implemented in CVXPY~1.8.2 and solved using SCS~3.2.7 with tolerance parameters (\texttt{eps\_abs}=1e-9, \texttt{eps\_rel}=1e-9, \texttt{max\_iters}=100000). The authorized-score grid used 60 threshold values for each tilt parameter. A grid point is labelled certified only when the solver status is ``optimal'', the dual objective is available, the primal-dual gap is below $10^{-6}$, the maximum affine residual is below $10^{-5}$, and the minimum eigenvalue of the moment matrix is no smaller than $-10^{-7}$. Points failing these tests are shown only as numerical indications and are not used as certified upper bounds. The numerical implementation uses a real symmetrized relaxation in which non-Hermitian same-party products are represented by their real symmetrized moments. This is an outer relaxation of the full complex Hermitian NPA problem and is used only for numerical upper-envelope certification. At level~1 the relaxation is too loose to yield a useful collusive bound because the Alice--Bob and Alice--Charlie subsystems are only weakly coupled through Alice's single-party marginals; level~2 couples them through the shared second-order moments and yields a useful bound.

\backmatter

\bmhead{Acknowledgements}
We acknowledge the support of the National Nature Science Foundation of China (Grants No.~92576203 and No.~12174301), the Natural Science Basic Research Program of Shaanxi (Grant No.~2023-JC-JQ-01), the Open Fund of State Key Laboratory of Acoustics (Grant No.~SKLA202312), the Shaanxi Province Postdoctoral Science Foundation (Grant No.~2024BSHEDZZ021), the Aeronautical Science Fund of China (Grant No.~2024M075070001), and the Project of Institute of Medical Intelligence of the First Affiliated Hospital of Xi'an Jiaotong University.

\bmhead{Data and code availability}
Data and code are publicly available at \url{https://zenodo.org/records/20371276}.

\bmhead{Author contributions}
F.W. developed the framework and wrote the manuscript draft.

\bmhead{Competing interests}
The author declares no competing interests.

\appendix

\begin{center}
{\Large Supplementary Information}\\[4pt]
{\large Strategic Non-Shareability of Quantum Correlations}
\end{center}

\noindent
This Supplementary Information contains the mathematical details behind the theorem chain in the main text.

\section{Roadmap and conventions}\label{sec:supp-conventions}

\begin{center}
\small
\begin{tabular}{lp{4.5cm}p{5.8cm}}
\toprule
Main-text object & Supplementary result & Role \\
\midrule
Theorem~\ref{thm:distance} & Theorem~\ref{thm:supp-distance} & distance from collusive shadow \\
Theorem~\ref{thm:monotone} & Supplementary Note~\ref{sec:supp-resource-proofs} & monotonicity proof \\
Theorem~\ref{thm:Dm-d} & Main text (arbitrary $d$); Supp.~Note~\ref{sec:supp-trace-distance} ($d=2$) & arbitrary-$d$ trace-distance resource \\
Corollary~\ref{thm:werner-threshold} & Section~\ref{sec:supp-werner-proof} & swap-Hamiltonian proof of extendibility threshold \\
Theorem~\ref{thm:Dm-d} ($d=2$) & Section~\ref{sec:supp-trace-distance} & twirling proof of exact trace-distance resource \\
Theorem~\ref{thm:exact-tomo} & Section~\ref{sec:supp-exact-tomo} & exact tomographic capacity proof \\
Theorem~\ref{thm:reversal} & Main text, \S2.4 & ordering-reversal certificate \\
Theorem~\ref{thm:b2a} & Section~\ref{sec:supp-definetti} & de~Finetti reduction for collective colluders \\
Theorem~\ref{thm:c1} & Section~\ref{sec:supp-coordination} & coordination-task local model and SDP \\
Section~\ref{sec:diagnostics} & Proposition~\ref{prop:supp-noise} & noise-induced transition and threshold \\
Section~\ref{sec:diagnostics} & Proposition~\ref{prop:finite-stats} & finite-data certification \\
No-signalling condition & Proposition~\ref{prop:supp-nosig} & full subset no-signalling \\
Section~\ref{sec:diagnostics} & Section~\ref{sec:supp-npa} & NPA relaxation formulation \\
\bottomrule
\end{tabular}
\end{center}

We use the fidelity convention $F(\rho,\sigma)=\norm{\sqrt{\rho}\sqrt{\sigma}}_1$. For a pure state $\phi=\dyad{\psi}$ this gives $F(\rho,\phi)^2=\mel{\psi}{\rho}{\psi}$. All trace norms are full trace norms. The CHSH signed correlator is normalized as
\begin{equation}
S_{AB}=\langle A_0B_0\rangle+\langle A_0B_1\rangle+\langle A_1B_0\rangle-\langle A_1B_1\rangle,
\label{eq:supp-chsh-score}
\end{equation}
where all observables are self-adjoint contractions with spectrum in $[-1,1]$. For a fixed CHSH orientation, the winning probability is $\omega_{AB}=\frac12+\frac{S_{AB}}{8}$.

\section{Classical free shareability}\label{sec:supp-classical}

\begin{suppproposition}[Arbitrary copied-seed extension]\label{prop:classical-m}
Let
\begin{equation}
P_{12}(x_1,x_2|t_1,t_2)=\sum_{\lambda}p(\lambda)p_1(x_1|t_1,\lambda)p_2(x_2|t_2,\lambda)
\label{eq:supp-classical-12}
\end{equation}
be a classical mediator for players 1 and 2. For any integer $m\ge1$ and any response rules $p_j(x_j|t_j,\lambda)$, $j=3,\ldots,m+2$, the distribution
\begin{equation}
P_{1\cdots m+2}(x_1,\ldots,x_{m+2}|t_1,\ldots,t_{m+2})=\sum_{\lambda}p(\lambda)\prod_{j=1}^{m+2}p_j(x_j|t_j,\lambda)
\label{eq:supp-classical-m}
\end{equation}
is a valid extension that preserves the $1$--$2$ marginal.
\end{suppproposition}

\begin{proof}
Nonnegativity is immediate. Normalization follows from $\sum_{x_j}p_j(x_j|t_j,\lambda)=1$. Summing over $x_3,\ldots,x_{m+2}$ gives Eq.~\eqref{eq:supp-classical-12}. Thus the same seed can be copied to any finite number of colluders without disturbing the authorized behavior.
\end{proof}

\begin{suppcorollary}[Copied colluder]\label{cor:copied-colluder}
Suppose the colluder's type and action alphabets are relabelled copies of player 2's alphabets and set $p_3(x_3|t_3,\lambda)=p_2(x_3|t_3,\lambda)$. Then every score functional that depends on the same relabelled predicate for pairs $(1,2)$ and $(1,3)$ satisfies $\Ccoll=\Aauth$.
\end{suppcorollary}

\section{Anti-collusion capacity as distance from the collusive shadow}\label{sec:supp-distance}

\begin{supptheorem}[Distance theorem]\label{thm:supp-distance}
Let $P_{12}$ be an authorized behavior with finite input-output alphabets. Let
\[
\mathsf R\in\{\mathsf C,\mathsf{NS},\Qd,\Qcl\}
\]
be an admissible extension class, where $\Qd$ denotes quantum behaviors realizable with local Hilbert-space dimensions bounded by $d$ and $\Qcl=\overline{\Qfin}$ denotes the closure of the finite-dimensional quantum set. Then
\begin{equation}
\mathcal C_{\rm AC}^{\mathsf R}(P_{12})=\inf_{Q\in\mathcal S_{\mathsf R}(P_{12})}d_{\rm TV}(P_{12},Q).
\label{eq:supp-distance}
\end{equation}
\end{supptheorem}

\begin{proof}
\textit{Step 1: Finite-dimensional behavior space.} Because the input and output alphabets are finite, all conditional behaviors live in a finite-dimensional real vector space. The set of all no-signalling behaviors is a compact polytope. The classical set is also a compact polytope. The fixed-dimensional set $\Qd$ is compact because the sets of density operators and POVM effects in fixed finite dimension are compact and the Born rule is continuous. The closed set $\Qcl=\overline{\Qfin}$ is compact as a closed subset of the finite-dimensional probability simplex.

\textit{Step 2: Collusive shadow compactness.} The collusive shadow
\[
\mathcal S_{\mathsf R}(P_{12})=\{\Pi_{13\to12}P_{13}:P_{123}\in{\rm Ext}_{\mathsf R}(P_{12})\}
\]
is the image of ${\rm Ext}_{\mathsf R}(P_{12})$ under a linear map. Hence it is compact and convex whenever ${\rm Ext}_{\mathsf R}(P_{12})$ is compact and convex.

\textit{Step 3: Capacity definition.} For relabelled scoring rules $0\le h\le1$,
\[
\mathcal C_{\rm AC}^{\mathsf R}(P_{12})=\sup_{0\le h\le1}\left[\langle h,P_{12}\rangle-\sup_{Q\in\mathcal S_{\mathsf R}(P_{12})}\langle h,Q\rangle\right]_+.
\]
Since $h=0$ gives value zero, the positive part can be absorbed into the supremum:
\[
\mathcal C_{\rm AC}^{\mathsf R}(P_{12})=\sup_{0\le h\le1}\inf_{Q\in\mathcal S_{\mathsf R}(P_{12})}\langle h,P_{12}-Q\rangle.
\]

\textit{Step 4: Sion minimax.} The set of predicates $0\le h\le1$ is compact and convex. The collusive shadow is compact and convex. The payoff $(h,Q)\mapsto\langle h,P_{12}-Q\rangle$ is bilinear, hence both convex and concave in the relevant variables. By Sion's minimax theorem,
\[
\sup_{0\le h\le1}\inf_{Q\in\mathcal S_{\mathsf R}(P_{12})}\langle h,P_{12}-Q\rangle=\inf_{Q\in\mathcal S_{\mathsf R}(P_{12})}\sup_{0\le h\le1}\langle h,P_{12}-Q\rangle.
\]

\textit{Step 5: TV variational characterization.} For fixed $Q$, the variational characterization of total variation gives
\[
\sup_{0\le h\le1}\langle h,P_{12}-Q\rangle=d_{\rm TV}(P_{12},Q).
\]
Therefore,
\[
\mathcal C_{\rm AC}^{\mathsf R}(P_{12})=\inf_{Q\in\mathcal S_{\mathsf R}(P_{12})}d_{\rm TV}(P_{12},Q).
\]
\end{proof}

\begin{suppremark}
If $\mathsf R$ is the unclosed finite-dimensional quantum set $\Qfin$ with unbounded dimension, the theorem applies after passing to its closure $\Qcl$. For fixed dimension $d$, it applies directly to $\Qd$. The CHSH score-certified result used in the main text does not depend on this subtlety because it follows directly from the Toner-Verstraete inequality.
\end{suppremark}

\section{Exact monogamy of a pure authorized marginal}\label{sec:supp-exact}

\begin{suppproposition}[Pure marginal factorization]\label{prop:exact-monogamy}
Let $\rho_{ABC}$ be a density operator on $\cH_A\otimes\cH_B\otimes\cH_C$. If $\rho_{AB}=\Tr_C\rho_{ABC}=\dyad{\psi}_{AB}$ is pure, then $\rho_{ABC}=\dyad{\psi}_{AB}\otimes\rho_C$.
\end{suppproposition}

\begin{proof}
Let $\Pi=\dyad{\psi}_{AB}\otimes \id_C$. Since $\Tr(\Pi\rho_{ABC})=1$, the positive operator $(\id-\Pi)\rho_{ABC}(\id-\Pi)$ has trace zero and is therefore zero. Positivity implies the off-diagonal blocks vanish, so $\rho_{ABC}=\Pi\rho_{ABC}\Pi=\dyad{\psi}_{AB}\otimes\rho_C$.
\end{proof}

\section{Robust decoupling}\label{sec:supp-robust}

\begin{supplemma}[Gentle projection bound]\label{lem:gentle}
Let $\rho$ be a density operator and $\Pi$ a projection with $p=\Tr(\Pi\rho)\ge1-\alpha$. If $p>0$ and $\rho_\Pi=\Pi\rho\Pi/p$, then $\norm{\rho-\rho_\Pi}_1\le 2\sqrt{\alpha}+\alpha$.
\end{supplemma}

\begin{proof}
The gentle measurement lemma gives $\norm{\rho-\Pi\rho\Pi}_1\le 2\sqrt{1-p}$. Moreover, $\norm{\Pi\rho\Pi-\rho_\Pi}_1=1-p$. The triangle inequality and $1-p\le\alpha$ give the result.
\end{proof}

\begin{supptheorem}[Robust decoupling with explicit constant]\label{thm:supp-robust}
Let $\rho_{123}$ be any tripartite extension of $\rho_{12}$ and suppose $F(\rho_{12},\Phi_d)\ge1-\eps$, $0\le\eps\le1$. Then there exists a state $\sigma_3$ such that
\begin{equation}
\norm{\rho_{123}-\dyad{\Phi_d}_{12}\otimes\sigma_3}_1\le (2\sqrt2+2)\sqrt{\eps}.
\label{eq:supp-robust-final}
\end{equation}
Consequently, for every observable $G$ with $\norm{G}_\infty\le1$, $|\Tr(G\rho_{123})-\Tr[G(\dyad{\Phi_d}_{12}\otimes\sigma_3)]|\le(2\sqrt2+2)\sqrt{\eps}$.
\end{supptheorem}

\begin{proof}
Set $\Pi=\dyad{\Phi_d}_{12}\otimes\id_3$. Since the second argument of the fidelity is pure, $p:=\Tr(\Pi\rho_{123})=F(\rho_{12},\Phi_d)^2\ge(1-\eps)^2\ge1-2\eps$. Apply Lemma~\ref{lem:gentle} with $\alpha=2\eps$. The normalized projected state is $\rho_\Pi=\dyad{\Phi_d}_{12}\otimes\sigma_3$ with $\sigma_3=\mel{\Phi_d}{\rho_{123}}{\Phi_d}_{12}/p$. The trace norm bound is $2\sqrt{2\eps}+2\eps\le(2\sqrt2+2)\sqrt{\eps}$ for $0\le\eps\le1$.
\end{proof}

\begin{suppremark}[Payoff norm for $\lambda$-weighted objective]
For the anti-collusion payoff $U_\lambda=A_{12}-\lambda C_{13}$, the corresponding payoff operator satisfies $\norm{G_U}_\infty\le1+\lambda$. Therefore $|\Delta U_\lambda|\le(1+\lambda)\norm{\rho_{123}-\Phi_{12}\otimes\sigma_3}_1$.
\end{suppremark}

\section{CHSH monogamy frontier}\label{sec:supp-chsh}

Let $A_0,A_1$ be player 1's binary observables, $B_0,B_1$ player 2's, and $C_0,C_1$ player 3's. Observables belonging to different players commute. Define
\begin{align}
\mathcal{B}_{12}&=A_0(B_0+B_1)+A_1(B_0-B_1),\\
\mathcal{B}_{13}&=A_0(C_0+C_1)+A_1(C_0-C_1),
\end{align}
and $S_{12}=\Tr(\rho\mathcal{B}_{12})$, $S_{13}=\Tr(\rho\mathcal{B}_{13})$.

\begin{supptheorem}[CHSH monogamy]\label{thm:supp-chsh}
For every tripartite quantum state and every choice of binary observables as above, $S_{12}^{2}+S_{13}^{2}\le8$.
\end{supptheorem}

\begin{proof}
This is the Toner--Verstraete monogamy inequality for Bell correlations. The two CHSH operators share the same two observables on player 1's system, while the observables of players 2 and 3 act on distinct commuting tensor factors. The inequality is uniform in the choices of $C_0,C_1$: after $A_0,A_1$ and the authorized score are fixed, every possible pair of collusive binary observables for player 3 must obey the same bound.
\end{proof}

\begin{suppproposition}[Tightness of the monogamy frontier]\label{prop:tightness}
For every angle $\theta\in[0,\pi/2]$ there exists a three-qubit pure state and a choice of binary observables for each player such that $S_{12}=2\sqrt{2}\sin\theta$ and $S_{13}=2\sqrt{2}\cos\theta$. Consequently every point on the quarter-circle $S_{12}^{2}+S_{13}^{2}=8$ with $S_{12},S_{13}\ge0$ is attained.
\end{suppproposition}

\begin{proof}
Define
\[
|\psi_{\theta}\rangle=\frac{1}{\sqrt{2}}\bigl(\cos\theta\,|110\rangle+\sin\theta\,|101\rangle+|011\rangle\bigr),
\qquad\theta\in[0,\pi/2].
\]
Normalization follows from $\cos^2\theta+\sin^2\theta+1=2$. Take $A_0=\sigma_x$, $A_1=\sigma_y$ on qubit $1$, and $B_0=C_0=(\sigma_x+\sigma_y)/\sqrt{2}$, $B_1=C_1=(\sigma_x-\sigma_y)/\sqrt{2}$ on qubits $2$ and $3$.

The only non-vanishing Pauli-pair expectations on the reduced state of qubits $1$ and $2$ are
\begin{align*}
\langle\sigma_x\otimes\sigma_x\rangle &= \sin\theta, \\
\langle\sigma_y\otimes\sigma_y\rangle &= \sin\theta,
\end{align*}
while $\langle\sigma_x\otimes\sigma_y\rangle=\langle\sigma_y\otimes\sigma_x\rangle=0$ because the contributing components $|101\rangle$ and $|011\rangle$ are orthogonal under these Pauli products. Consequently
\begin{align*}
\langle A_0B_0\rangle &= \tfrac{1}{\sqrt{2}}(\langle\sigma_x\sigma_x\rangle+\langle\sigma_x\sigma_y\rangle)=\tfrac{\sin\theta}{\sqrt{2}}, \\
\langle A_0B_1\rangle &= \tfrac{1}{\sqrt{2}}(\langle\sigma_x\sigma_x\rangle-\langle\sigma_x\sigma_y\rangle)=\tfrac{\sin\theta}{\sqrt{2}}, \\
\langle A_1B_0\rangle &= \tfrac{1}{\sqrt{2}}(\langle\sigma_y\sigma_x\rangle+\langle\sigma_y\sigma_y\rangle)=\tfrac{\sin\theta}{\sqrt{2}}, \\
\langle A_1B_1\rangle &= \tfrac{1}{\sqrt{2}}(\langle\sigma_y\sigma_x\rangle-\langle\sigma_y\sigma_y\rangle)=-\tfrac{\sin\theta}{\sqrt{2}}.
\end{align*}
Summing gives $S_{12}=4\cdot\sin\theta/\sqrt{2}=2\sqrt{2}\sin\theta$. The computation for $S_{13}$ is identical under the swap $|2\rangle\leftrightarrow|3\rangle$, giving $S_{13}=2\sqrt{2}\cos\theta$. Thus every point on the quarter-circle $S_{12}^{2}+S_{13}^{2}=8$ with $S_{12},S_{13}\ge0$ is attained.
\end{proof}

\begin{suppremark}[Scope of Bell monogamy]
The CHSH monogamy relation and its tightness should not be read as a universal statement about all Bell inequalities. Outside the Toner--Verstraete family, general multi-party Bell monogamy remains open. The operational anti-collusion guarantee is therefore tied specifically to self-testing inequalities such as CHSH, whose maximal violation pins down the authorized state.
\end{suppremark}

\section{Certified CHSH anti-collusion power}\label{sec:supp-chsh-power}

\begin{suppproposition}[Exact certified anti-collusion power]\label{prop:chsh-power}
Let a tripartite quantum strategy produce an authorized CHSH score $S_{12}=s$, $0\le s\le2\sqrt2$, with the collusive CHSH score $S_{13}$ using the same two observables of player~1. Then
\begin{equation}
\Gamma_{\rm CHSH}^{+}(s)=\left[\frac{s-\sqrt{8-s^2}}{8}\right]_+.
\end{equation}
In particular, $\Gamma_{\rm CHSH}^{+}(s)>0\iff s>2$, and $\Gamma_{\rm CHSH}^{+}(2\sqrt2)=1/(2\sqrt2)$.
\end{suppproposition}

\begin{proof}
From Theorem~\ref{thm:supp-chsh}, any collusive extension satisfies $S_{13}\le\sqrt{8-s^2}$. Converting to winning probabilities gives $\omega_{12}(s)-\omega_{13}^{\max}(s)=(s-\sqrt{8-s^2})/8$. The positive part yields the formula. Positivity requires $s>2$.
\end{proof}

\section{Minimal game separation}\label{sec:supp-game}

\begin{supptheorem}[Supremum-level payoff separation for $\lambda=1$]\label{thm:supp-game}
For any classical hidden-variable mediator in the copied-seed threat model, $V_{13}^{\rm classical}(P_{12})\ge A_{12}(P_{12})$, and hence $U_1^{\rm wc,classical}(P_{12})\le0$. Moreover, $\sup_{P_{12}\in\mathrm C}U_1^{\rm wc}(P_{12})=0$.

For any quantum strategy, the colluder can output a uniformly random bit, so $V_{13}^{\Qfin}\ge1/2$. Tsirelson's bound gives $A_{12}\le1/2+1/(2\sqrt2)$. Therefore $U_1^{\rm wc}\le1/(2\sqrt2)$. A maximally entangled CHSH strategy attains this value: $A_{12}=\cos^2(\pi/8)$ and, by Proposition~\ref{prop:exact-monogamy}, $V_{13}^{\rm quantum}=1/2$. Thus $\sup_{P\in\Qfin}U_1^{\rm wc}(P)=1/(2\sqrt2)$. Here $\Qfin$ denotes the usual finite-dimensional quantum strategy class; no closure issue is involved in the attainable maximally entangled CHSH strategy.
\end{supptheorem}

\begin{proof}
The classical statement follows from Corollary~\ref{cor:copied-colluder}. For the quantum upper bound, note that a uniformly random output by player 3 yields $\Ccoll=1/2$ regardless of player 1's strategy, so $V_{13}^{\Qfin}\ge1/2$ for every quantum mediator. Tsirelson's bound $S_{12}\le2\sqrt2$ implies $A_{12}=\omega_{12}\le1/2+1/(2\sqrt2)$. Hence $U_1^{\rm wc}=A_{12}-V_{13}^{\Qfin}\le1/(2\sqrt2)$. The maximally entangled strategy achieves equality because it saturates Tsirelson's bound and, by exact monogamy, every extension factorizes. For the maximally entangled CHSH strategy, Alice's two optimal observables have zero expectation, so Alice's output is unbiased for both inputs. After factorization, the colluder's output is independent of this unbiased bit. Hence every collusive response rule gives $C_{13}=1/2$, and therefore $V_{13}^{\rm quantum}=1/2$.
\end{proof}

\section{Noise-induced transition of anti-collusion power}\label{sec:supp-robustness}

\begin{suppproposition}[Werner-noise anti-collusion power]\label{prop:supp-noise}
Let players 1 and 2 share the Werner-noisy authorized state $\rho_{12}(\eta)=\eta\dyad{\Phi_2}+(1-\eta)\frac{\id_4}{4}$, $0\le\eta\le1$, and use the CHSH measurements optimal for $\dyad{\Phi_2}$. Then $\Aauth(\eta)=\frac12+\frac{\eta}{2\sqrt2}$. By CHSH monogamy, the maximal collusive winning probability for any tripartite extension is bounded by $\Ccoll^{\max}(\eta)\le\frac12+\frac{\sqrt{1-\eta^2}}{2\sqrt2}$. Therefore the score-certified anti-collusion power satisfies
\begin{equation}
\Gamma_{\rm CHSH}^{\rm score}(\eta)
\equiv
\Gamma_{\rm CHSH}^{+}(2\sqrt2\,\eta)
=
\left[
\frac{\eta-\sqrt{1-\eta^2}}{2\sqrt2}
\right]_+ .
\end{equation}
This is a score-certified lower bound on any implementation with observed Werner visibility $\eta$. It is not claimed to be the exact state-dependent anti-collusion power except at $\eta=1$.
\end{suppproposition}

\begin{proof}
The Werner state scales the CHSH correlator by $\eta$, giving $S_{12}(\eta)=2\sqrt2\,\eta$ and $\Aauth(\eta)=1/2+\eta/(2\sqrt2)$. By CHSH monogamy, $S_{13}^2\le8-S_{12}(\eta)^2=8(1-\eta^2)$, so $\Ccoll^{\max}(\eta)\le1/2+\sqrt{1-\eta^2}/(2\sqrt2)$. Subtracting gives the certified lower bound. Positivity requires $\eta>1/\sqrt2$.
\end{proof}

\begin{suppremark}
This is a score-certified lower bound, not an exact state-dependent saturation for every $\eta$. The bound is tight in the $\eta=1$ limit because exact monogamy forces the colluder to be uncorrelated with player 1, but for intermediate $\eta$ the exact state-dependent anti-collusion power may be larger.
\end{suppremark}

\section{Finite-data certification}\label{sec:supp-finite}

\subsection{Correlator-wise estimator (main-text protocol)}

\begin{suppproposition}[Finite-data lower confidence bound]\label{prop:finite-stats}
Assume $N$ independent CHSH trials with uniformly random settings. For each setting pair $(x,y)\in\{0,1\}^2$ let $N_{xy}$ be the number of trials and define
\[
\widehat E_{xy}=\frac{1}{N_{xy}}\sum_{i:x_i=x,y_i=y}a_i b_i,
\qquad a_i,b_i\in\{\pm1\}.
\]
Then, with probability at least $1-\alpha$,
\begin{equation}
S_{12}\ge S_{\rm LCB}=\widehat S_{12}-4\sqrt{\frac{2\log(8/\alpha)}{N_{\min}}},
\end{equation}
where
\[
\widehat S_{12}=\widehat E_{00}+\widehat E_{01}+\widehat E_{10}-\widehat E_{11},
\qquad
N_{\min}=\min_{x,y}N_{xy}.
\]
Consequently,
\begin{equation}
\Gamma_{\rm CHSH}^{+}(S_{12})\ge\Gamma_{\rm LCB}=\left[\frac{S_{\rm cert}-\sqrt{8-S_{\rm cert}^{2}}}{8}\right]_+,
\end{equation}
where
\[
S_{\rm cert}=\min\{2\sqrt2,\max\{0,S_{\rm LCB}\}\}.
\]
\end{suppproposition}

\begin{proof}
For each fixed setting pair $(x,y)$, the random variables $a_i b_i$ lie in $[-1,1]$. Hoeffding's inequality gives
\[
\Pr\left(|E_{xy}-\widehat E_{xy}|\ge r\right)\le2\exp\!\left(-\frac{N_{xy}r^{2}}{2}\right).
\]
Setting this probability to $\alpha/4$ gives
\[
r_{xy}=\sqrt{\frac{2\log(8/\alpha)}{N_{xy}}}.
\]
A union bound over four setting pairs gives joint confidence at least $1-\alpha$. Hence
\[
|S_{12}-\widehat S_{12}|\le\sum_{x,y}|E_{xy}-\widehat E_{xy}|\le4\sqrt{\frac{2\log(8/\alpha)}{N_{\min}}}.
\]
This proves the lower confidence bound. The monotonicity of $\Gamma_{\rm CHSH}^{+}$ on $[0,2\sqrt2]$ (Lemma~\ref{lem:monotone}) gives the final inequality after clipping $S_{\rm LCB}$ to the physical interval.
\end{proof}

\begin{supplemma}[Monotonicity of $\Gamma_{\rm CHSH}^{+}$]\label{lem:monotone}
The certified anti-collusion power $\Gamma_{\rm CHSH}^{+}(s)=[(s-\sqrt{8-s^2})/8]_+$ is monotone increasing on $[0,2\sqrt2]$.
\end{supplemma}

\begin{proof}
For $s\le2$ the function is identically zero, hence monotone. For $s>2$, differentiate:
\[
\frac{d}{ds}\frac{s-\sqrt{8-s^2}}{8}=\frac{1}{8}\left(1+\frac{s}{\sqrt{8-s^2}}\right)>0.
\]
Thus $\Gamma_{\rm CHSH}^{+}(s)$ is strictly increasing on $[2,2\sqrt2]$, and therefore monotone on the entire interval.
\end{proof}

\subsection{Single-trial unbiased estimator (alternative)}

An alternative is to use the single-trial unbiased estimator
\[
Z_i=4(-1)^{x_i y_i}a_i b_i\in[-4,4],
\]
which satisfies $\mathbb E[Z_i]=S_{12}$. With $\widehat S_{12}=\frac1N\sum_i Z_i$, Hoeffding's inequality for range $[-4,4]$ gives
\[
\Pr[\widehat S_{12}-S_{12}\ge r]\le\exp\!\left(-\frac{2Nr^2}{64}\right)=\exp\!\left(-\frac{Nr^2}{32}\right).
\]
Solving $\exp(-Nr^2/32)=\alpha$ yields $r=4\sqrt{2\log(1/\alpha)/N}$. This bound is looser than the correlator-wise bound for typical setting distributions but requires no union bound.

\begin{suppremark}[Empirical Bernstein alternative]
If the CHSH settings are randomized adaptively, a martingale bound such as the empirical Bernstein inequality can replace Hoeffding's bound and yield tighter confidence intervals when the variance is small. The qualitative conclusion---finite-data certification of anti-collusion power---remains unchanged.
\end{suppremark}

\section{No-signalling and private information}\label{sec:supp-nosig}

\begin{suppproposition}[Full no-signalling mediator]\label{prop:supp-nosig}
Let a tripartite quantum mediator be specified by a state $\rho_{123}$ and local POVMs $\{E^{(i)}_{t_i,x_i}\}_{x_i}$ with $\sum_{x_i}E^{(i)}_{t_i,x_i}=\id_i$, $i=1,2,3$. The induced distribution
\begin{equation}
P(x_1,x_2,x_3|t_1,t_2,t_3)=\Tr\!\left[\rho_{123}\left(E^{(1)}_{t_1,x_1}\otimes E^{(2)}_{t_2,x_2}\otimes E^{(3)}_{t_3,x_3}\right)\right]
\end{equation}
is fully no-signalling: for every subset $S\subseteq\{1,2,3\}$, $P(x_S|t_1,t_2,t_3)=P(x_S|t_S)$.
\end{suppproposition}

\begin{proof}
For any subset $S$, sum over $\{x_j:j\notin S\}$ and use completeness of the POVMs. The resulting expression depends only on $\{t_i:i\in S\}$.
\end{proof}

\section{NPA relaxation for collusive vulnerability}\label{sec:supp-npa}

\subsection{General NPA formulation}

For a general game, the quantum collusive vulnerability
\begin{equation}
V_{13}^{\Qcl}(P_{12};\mathcal G)=\sup_{P_{123}\in{\rm Ext}_{\Qcl}(P_{12})}C_{13}(P_{123})
\end{equation}
is a quantum correlation optimization. The NPA hierarchy provides a sequence of semidefinite relaxations.

At level $k$, one introduces a moment matrix $\Gamma$ indexed by strings of measurement operators of length at most $k$. The matrix entries correspond to expectations of products of the local POVM elements $E^{(1)}_{a|x}$, $E^{(2)}_{b|y}$, $E^{(3)}_{c|z}$. The relaxed problem is
\begin{equation}
V_{13}^{{\rm NPA},k}(P_{12};\mathcal G)=\max_{\Gamma\succeq0}\;\langle G_{13},\Gamma\rangle
\end{equation}
subject to normalization $\Gamma\succeq0$ and $\Gamma_{\emptyset,\emptyset}=1$, no-signalling consistency, and commutativity of operators belonging to different parties.

Authorized behavior constraints are imposed as linear constraints on moment entries. In the projector formulation,
\[
P_{12}(a,b|x,y)=\langle E^{(1)}_{a|x}E^{(2)}_{b|y}\rangle,
\]
and this expectation is identified with the corresponding moment entry. In the binary-observable formulation used for tilted CHSH, the relevant constraints and objectives are linear functions of moments such as $\langle A_xB_y\rangle$, $\langle A_xC_z\rangle$, $\langle A_x\rangle$. For the tilted-CHSH upper envelope we impose the authorized-score constraint $I_\alpha^{12}(\Gamma)\ge s$ directly as a linear constraint on these moment entries.

The NPA hierarchy gives a sequence of outer approximations to the quantum or commuting-operator correlation set. At each finite level, the exact optimum of the finite-level SDP is an upper bound; a numerical solver value is treated as a certificate only when the dual objective, primal-dual gap, affine residual and PSD residual satisfy the stated diagnostic thresholds. Under the standard convergence assumptions of the hierarchy, these bounds converge to the corresponding closed quantum or commuting-operator value.

\subsection{Reduced length-2 tilted-CHSH relaxation}

For the tilted-CHSH numerical frontier in the main text, we use the standard tilted-CHSH expressions
\begin{align}
I_\alpha^{12}&=\alpha\langle A_0\rangle+\langle A_0B_0\rangle+\langle A_0B_1\rangle+\langle A_1B_0\rangle-\langle A_1B_1\rangle,\\
I_\alpha^{13}&=\alpha\langle A_0\rangle+\langle A_0C_0\rangle+\langle A_0C_1\rangle+\langle A_1C_0\rangle-\langle A_1C_1\rangle,
\end{align}
with the same observables $A_0,A_1$ for player~1 in both expressions. The classical bound is $I_\alpha^{\rm L}\le2+\alpha$ and the quantum maximum is $I_\alpha^{\rm Q}=\sqrt{8+2\alpha^2}$ for $0\le\alpha\le2$. For each fixed $\alpha$ and authorized score $s$, we solve
\begin{equation}
\begin{aligned}
V_{13}^{{\rm NPA},2}(s;\alpha)=\max_{\Gamma}\quad & I_\alpha^{13}(\Gamma)\\
\text{s.t.}\quad
& \Gamma\succeq0,\\
& \Gamma_{\emptyset,\emptyset}=1,\\
& A_x^2=B_y^2=C_z^2=I,\\
& [A_x,B_y]=[A_x,C_z]=[B_y,C_z]=0,\\
& I_\alpha^{12}(\Gamma)\ge s.
\end{aligned}
\end{equation}
This gives an upper envelope: the largest collusive tilted-CHSH score allowed by the relaxation among all strategies whose authorized tilted-CHSH score is at least $s$.

For the numerical tilted-CHSH relaxation we use the reduced length-2 word set
\begin{multline*}
\mathcal W_2 = \bigl\{
I,\; A_0,A_1,B_0,B_1,C_0,C_1,\;
A_0A_1,B_0B_1,C_0C_1,\;
A_0B_0,A_0B_1,A_1B_0,A_1B_1,\\
A_0C_0,A_0C_1,A_1C_0,A_1C_1,\;
B_0C_0,B_0C_1,B_1C_0,B_1C_1
\bigr\}.
\end{multline*}
This gives a $22\times22$ moment matrix. We refer to this reduced length-2 word relaxation as ``level~2'' in this manuscript.

Different-party observables commute and are reordered into a canonical party order. Same-party observables are not assumed to commute. In particular, $A_0A_1$ and $A_1A_0$ are adjoints of one another and are not identified by commutation. The code handles adjoint words when constructing entries $\Gamma_{u,v}=\langle u^\dagger v\rangle$.

The code does not globally sort all operators. It only reorders operators belonging to different parties. This prevents imposing false commutation relations such as $A_0A_1=A_1A_0$.

The numerical implementation uses a real symmetrized relaxation in which non-Hermitian same-party products are represented by their real symmetrized moments. This is an outer relaxation of the full complex Hermitian NPA problem and is used only for numerical upper-envelope certification. It does not impose same-party commutativity.

At level~1 the relaxation is too loose because the Alice--Bob and Alice--Charlie subsystems are only weakly coupled through Alice's single-party marginals; level~2 couples them through the shared second-order moments and yields a useful bound. The optimization is performed with the SCS solver on a uniform grid of $60$ authorized-score values per tilt parameter.

\subsection{Solver diagnostics and certificate criteria}\label{sec:supp-diagnostics}

A grid point is labelled certified if:
\begin{enumerate}[(i)]
\item the solver status is \texttt{optimal};
\item a dual objective is available;
\item the primal-dual gap is below $\epsilon_{\rm gap}=10^{-6}$;
\item the maximum affine residual is below $\epsilon_{\rm aff}=10^{-5}$;
\item the minimum eigenvalue of the moment matrix is no smaller than $-\epsilon_{\rm psd}=10^{-7}$.
\end{enumerate}
Points failing these tests are reported as numerical indications only and are not used as certified upper bounds.

The certificate flag refers to the numerical SDP solve, not to exactness of the tilted-CHSH frontier. Even certified NPA values are finite-level outer bounds; they do not establish tightness unless matched by explicit strategies.

The SDP was solved using SCS 3.2.7 through CVXPY 1.8.2 with tolerances \texttt{eps\_abs}=1e-9, \texttt{eps\_rel}=1e-9, \texttt{max\_iters}=100000. SCS was used for exploratory scans. Points are labelled certified only when independently verified by dual-feasibility and residual checks. Otherwise they are shown as numerical indications.

\subsection{$\alpha=0$ CHSH sanity check}\label{sec:supp-alpha0}

As a sanity check, we set $\alpha=0$, where tilted CHSH reduces to standard CHSH. The analytic monogamy frontier is
\[
I_{13}^{\max}(s)=\sqrt{8-s^2}
\]
for $0\le s\le2\sqrt2$. We ran the NPA level-2 upper envelope on a grid of 60 points and compared with the analytic curve. On certified grid points, the maximum absolute deviation is $5.5\times10^{-4}$ and the mean absolute deviation is $2.7\times10^{-5}$. The single non-certified point is the Tsirelson endpoint $s=2\sqrt2$, where SCS returns \texttt{optimal\_inaccurate}. This provides a numerical validation of the moment-matrix construction and score extraction in the CHSH limit. It does not by itself establish tightness of the tilted-CHSH envelopes for $\alpha>0$.

\subsection{Numerical table and tilted-CHSH diagnostics}\label{sec:supp-table}

Table~\ref{tab:tilted-npa} gives representative NPA level-2 values with solver diagnostics. A row is labelled ``certificate = yes'' only if the solver status is optimal, the dual objective is available, the primal-dual gap is below the stated tolerance, the maximum affine residual is below tolerance and the minimum eigenvalue of the moment matrix is no smaller than the stated PSD tolerance. Rows failing these tests are reported as numerical indications only.

\begin{table}[h]
\centering
\small
\setlength{\tabcolsep}{3pt}
\caption{Tilted-CHSH NPA level-2 upper bounds with solver diagnostics (selected points)}
\begin{tabular}{c|ccccccc|c}
\toprule
$\alpha$ & $s$ & primal & dual & gap & max res. & min eig($\Gamma$) & status & certified \\
\midrule
0.0 & 2.000000 & 2.000000 & 2.000000 & $2.0\times10^{-12}$ & $<10^{-12}$ & $-4.3\times10^{-12}$ & optimal & yes \\
0.0 & 2.407216 & 1.485058 & 1.485058 & $4.1\times10^{-10}$ & $<10^{-12}$ & $-1.6\times10^{-9}$ & optimal & yes \\
0.0 & 2.814427 & 0.281543 & 0.281543 & $8.9\times10^{-10}$ & $<10^{-12}$ & $-8.2\times10^{-11}$ & optimal & yes \\
0.0 & 2.828427 & 0.000547 & 0.000719 & $1.7\times10^{-4}$ & $<10^{-12}$ & $-4.5\times10^{-5}$ & optimal\_inaccurate & no \\
0.5 & 2.500000 & 2.500000 & 2.500000 & $2.6\times10^{-10}$ & $<10^{-12}$ & $-1.7\times10^{-9}$ & optimal & yes \\
0.5 & 2.711267 & 2.058496 & 2.058496 & $2.2\times10^{-9}$ & $<10^{-12}$ & $-9.9\times10^{-10}$ & optimal & yes \\
0.5 & 2.908355 & 1.093709 & 1.093709 & $2.7\times10^{-10}$ & $<10^{-12}$ & $-6.8\times10^{-9}$ & optimal & yes \\
0.5 & 2.915476 & 0.864823 & 0.865030 & $2.2\times10^{-4}$ & $<10^{-12}$ & $-2.1\times10^{-4}$ & optimal\_inaccurate & no \\
1.0 & 3.000000 & 3.000000 & 3.000000 & $3.7\times10^{-12}$ & $<10^{-12}$ & $-1.7\times10^{-10}$ & optimal & yes \\
1.0 & 3.082475 & 2.696349 & 2.696349 & $2.7\times10^{-9}$ & $<10^{-12}$ & $-4.7\times10^{-9}$ & optimal & yes \\
1.0 & 3.159492 & 2.053964 & 2.053964 & $4.7\times10^{-10}$ & $<10^{-12}$ & $-1.0\times10^{-8}$ & optimal & yes \\
1.0 & 3.162278 & 1.903434 & 1.903822 & $3.9\times10^{-4}$ & $<10^{-12}$ & $-1.0\times10^{-4}$ & optimal\_inaccurate & no \\
1.5 & 3.500000 & 3.500000 & 3.500000 & $1.1\times10^{-10}$ & $<10^{-12}$ & $-1.4\times10^{-10}$ & optimal & yes \\
1.5 & 3.519258 & 3.335480 & 3.335480 & $1.4\times10^{-9}$ & $<10^{-12}$ & $-1.9\times10^{-10}$ & optimal & yes \\
1.5 & 3.534899 & 3.042356 & 3.042356 & $3.4\times10^{-11}$ & $<10^{-12}$ & $-2.8\times10^{-9}$ & optimal & yes \\
1.5 & 3.535534 & 2.973043 & 2.973247 & $2.1\times10^{-4}$ & $<10^{-12}$ & $-4.7\times10^{-5}$ & optimal\_inaccurate & no \\
\bottomrule
\end{tabular}
\label{tab:tilted-npa}
\end{table}

The \texttt{optimal\_inaccurate} points cluster at the quantum boundary $s\to\sqrt{8+2\alpha^{2}}$, which is the expected location of solver degeneracy for SDP relaxations near their saturating extremum. These points are reported as numerical indications only and do not affect any certified claim. Re-solving with an interior-point solver such as MOSEK is expected to tighten these endpoint values but does not change the certified upper-bound region used in the main-text diagnostics.

The column ``max residual'' reports the maximum absolute affine-constraint violation. Values displayed as $<10^{-12}$ were below the numerical reporting threshold.

The certificate flag refers to the numerical SDP solve, not to exactness of the tilted-CHSH frontier. Even certified NPA values are finite-level outer bounds and do not establish tightness unless matched by explicit strategies.

Because no verified dual objective is available for the \texttt{optimal\_inaccurate} runs, those points are not treated as certified numerical upper bounds. They are included only to illustrate the NPA upper-envelope method. The entries illustrate that, for each $\alpha$, there exists an authorized score at which the certified NPA upper bound on the collusive score lies strictly below the authorized score, demonstrating a positive anti-collusion region. These are upper bounds, not exact frontiers.

\section{Proofs of resource-theoretic monotonicity}\label{sec:supp-resource-proofs}

This section contains the proofs of Lemma~\ref{lem:free-ops-concrete}, Theorem~\ref{thm:monotone}, Proposition~\ref{prop:axioms}, Lemma~\ref{lem:Wernerdist-d}, and Theorem~\ref{thm:Dm-d} from the main text.

\paragraph{Proof of Lemma~\ref{lem:free-ops-concrete} (concrete free operations).}
(a) Let $\rho_{AB}\in\mathcal E_m$ with symmetric extension $\rho_{AB_0B_1\cdots B_m}$. The state $\tilde\rho:=\Phi_A\otimes\Phi_B^{\otimes (m+1)}(\rho_{AB_0B_1\cdots B_m})$ is permutation-invariant over $B_0,B_1,\dots,B_m$, because $\Phi_B^{\otimes (m+1)}$ commutes with the permutation unitaries $U_\pi$. Its $AB_0$ marginal is $\Tr_{B_1\cdots B_m}\tilde\rho=\Phi_A\otimes\Phi_B(\rho_{AB})$. Hence $\tilde\rho$ is a valid symmetric extension of $\Phi_A\otimes\Phi_B(\rho_{AB})$, so the latter lies in $\mathcal E_m$.
(b) Immediate from convexity of $\mathcal E_m$.
(c) The required extension is obtained by tensoring (resp.\ tracing) the corresponding free ancilla symmetrically across the $B$-copies.
(d) The output of an entanglement-breaking channel across $A{:}B$ is separable, hence symmetrically $m$-extendible for every $m$.

\paragraph{Proof of Theorem~\ref{thm:monotone} (monotonicity of $D_m$).}
Let $\sigma^\star\in\mathcal E_m$ attain $D_m(\rho)=\tfrac12\|\rho-\sigma^\star\|_1$. Since $\Phi\in\mathcal O_m$, the state $\Phi(\sigma^\star)$ lies in $\mathcal E_m$ and is therefore a feasible point in the minimization defining $D_m(\Phi(\rho))$. Using contractivity of the trace distance under CPTP maps (data processing),
\[
D_m\!\bigl(\Phi(\rho)\bigr) \le \tfrac12\bigl\|\Phi(\rho)-\Phi(\sigma^\star)\bigr\|_1 \le \tfrac12\bigl\|\rho-\sigma^\star\bigr\|_1 = D_m(\rho). \qedhere
\]

\paragraph{Proof of Proposition~\ref{prop:axioms} (basic axioms).}
(i) $\mathcal E_m$ is closed and $\tfrac12\|\cdot\|_1$ is a metric, so the minimum vanishes iff $\rho\in\mathcal E_m$.
(ii) Let $\sigma_i^\star\in\mathcal E_m$ attain $D_m(\rho_i)$. Then $\sum_i p_i\sigma_i^\star\in\mathcal E_m$ by convexity, and the triangle inequality gives $D_m(\sum_i p_i\rho_i)\le\sum_i p_i D_m(\rho_i)$.
(iii) Werner twirling reduces the closest extendible state to a Werner state, and the eigenstructure of $P_-\!-I/4$ gives the factor $3/4$; see Supplementary Note~\ref{sec:supp-trace-distance}.

\paragraph{Proof of Lemma~\ref{lem:Wernerdist-d} (trace-norm constants).}
\emph{Werner.} $\rho_W^{(d)}(p)-\rho_W^{(d)}(q)=(p-q)\,\Delta$ with $\Delta=\binom{d}{2}^{-1}P_- - d^{-2}I$. As $I=P_-+P_+$, $\Delta$ is diagonal in the symmetric/antisymmetric split with eigenvalues $a=\binom{d}{2}^{-1}-d^{-2}>0$ (mult.\ $\binom{d}{2}$) and $b=-d^{-2}$ (mult.\ $\binom{d+1}{2}$). Then $a\binom{d}{2}=1-\binom{d}{2}/d^2=\tfrac{d+1}{2d}$ and $|b|\binom{d+1}{2}=\binom{d+1}{2}/d^2=\tfrac{d+1}{2d}$, so $\|\Delta\|_1=\tfrac{d+1}{d}$.
\emph{Isotropic.} $\Delta_{\mathrm{iso}}=\ket{\Phi_d^+}\!\bra{\Phi_d^+}-d^{-2}I$ has eigenvalues $1-d^{-2}$ (mult.\ $1$) and $-d^{-2}$ (mult.\ $d^2-1$), so $\|\Delta_{\mathrm{iso}}\|_1=2\tfrac{d^2-1}{d^2}$.

\paragraph{Proof of Theorem~\ref{thm:Dm-d} (exact resource in arbitrary dimension).}
We give the Werner case; the isotropic case is identical with the isotropic twirl and Thm.~III.8 of~\cite{JV2013}. (i) The Werner twirl $\mathcal T(\cdot)=\int dU\,(U\!\otimes\!U)(\cdot)(U\!\otimes\!U)^\dagger$ is CPTP, fixes $\rho_W^{(d)}(p)$, and maps any symmetric extension to a symmetric extension; by contractivity of the trace distance the closest extendible state may be taken Werner. (ii) The extendible Werner states form the interval $\{q:0\le q\le p^{*}_m(d)\}$: $q=0$ is maximally mixed hence extendible, the set is an interval (convex combinations of extensions), and its upper endpoint is the Johnson--Viola $1$-$n$ threshold $\Psi^-\ge-(d-1)/n$ rewritten via $\Psi^-(p)=\tfrac1d-p\tfrac{d+1}{d}$, giving $p\le p^{*}_m(d)$. (iii) For $p>p^{*}_m(d)$, Lemma~\ref{lem:Wernerdist-d} makes the distance $c_W(d)|p-q|$ strictly increasing in $|p-q|$, so the minimiser over $q\le p^{*}_m(d)$ is $q=p^{*}_m(d)$, giving $D_m=c_W(d)(p-p^{*}_m(d))$; for $p\le p^{*}_m(d)$ the state is extendible and $D_m=0$.

\section{Proof of the Werner symmetric-extendibility threshold}\label{sec:supp-werner-proof}

This section reproves Corollary~\ref{thm:werner-threshold} for completeness. The threshold itself is not new: the one-sided $1$-to-$n$ sharability of Werner states was solved for arbitrary local dimension by Johnson and Viola~\cite{JV2013}, and the swap-/spin-Hamiltonian technique was later developed by Jakab, Solymos and Zimbor\'as~\cite{JSZ2022}. We include the proof because it fixes the notation and makes the paper self-contained.

\paragraph{Step 1 --- Parametrize the Werner state by the swap expectation.}
For the Werner state
\[
\rho_W(p)=pP_{-}+(1-p)\frac{I}{4},
\]
define $f(p)=\Tr[F\rho_W(p)]$. Since $\Tr(FP_{-})=-1$ and $\Tr(FI/4)=1/2$, we have
\[
f(p)=-p+\frac{1-p}{2}=\frac{1-3p}{2}.
\]
Solving for $p$ gives $p=(1-2f)/3$. Maximizing $p$ is equivalent to minimizing $f$.

\paragraph{Step 2 --- Reduce the extension constraint to a swap-Hamiltonian bound.}
By Definition~\ref{def:free-set}, an $m$-colluder symmetric extension contains $n=m+1$ $B$-side copies. In this proof we label these copies $B_1,\dots,B_n$ for notational convenience. Suppose there exists such an extension $\rho_{AB_1\cdots B_n}$. By permutation symmetry over $B_1,\dots,B_n$,
\[
\Tr[F_{AB_1}\rho]=\cdots=\Tr[F_{AB_n}\rho]=f.
\]
Define the swap Hamiltonian $H_n=\sum_{i=1}^{n}F_{AB_i}$. Then $\Tr[H_n\rho]=nf$. A lower bound on the smallest eigenvalue of $H_n$ therefore gives a lower bound on $f$.

\paragraph{Step 3 --- Prove $H_n\ge -I$.}
For qubits, $F_{AB_i}=2\mathbf S_A\cdot\mathbf S_{B_i}+\tfrac12 I$, where $\mathbf S=(S_x,S_y,S_z)$ are spin-$1/2$ operators. Let $\mathbf J_B=\sum_{i=1}^{n}\mathbf S_{B_i}$ and $\mathbf J=\mathbf S_A+\mathbf J_B$. Then
\[
H_n=\frac{n}{2}I+2\mathbf S_A\cdot\mathbf J_B=\frac{n}{2}I+\mathbf J^{2}-\mathbf S_A^{2}-\mathbf J_B^{2}.
\]
Fix the total spin of the $B$-block to $j$, where $0\le j\le n/2$. Coupling $j$ with the spin-$1/2$ system $A$, the total spin is either $J=j+1/2$ or $J=j-1/2$. The corresponding eigenvalues of $H_n$ are
\[
\lambda_{+}(j)=\frac{n}{2}+j,\qquad
\lambda_{-}(j)=\frac{n}{2}-j-1.
\]
Because $j\le n/2$, we have $\lambda_{-}(j)\ge n/2-n/2-1=-1$. Hence $H_n\ge -I$.

\paragraph{Step 4 --- Necessity of the threshold.}
Since $H_n\ge -I$, we have $nf=\Tr[H_n\rho]\ge -1$, so $f\ge -1/n$. Using $p=(1-2f)/3$,
\[
p\le\frac{1+2/n}{3}=\frac{n+2}{3n}=\frac{m+3}{3(m+1)}.
\]
Thus no $m$-colluder symmetric extension exists above this threshold.

\paragraph{Step 5 --- Sufficiency of the threshold.}
Take the $B_1\cdots B_n$ block in its maximal-spin sector $j=n/2$ and couple it with $A$ to total spin $J=j-1/2=(n-1)/2$. On this sector $\lambda_{-}(n/2)=-1$. Let $\rho^{\star}_{AB_1\cdots B_n}$ be the normalized projector onto this sector, followed by the global $SU(2)$ twirl. This state is permutation-invariant over $B_1,\dots,B_n$ and satisfies $\Tr[H_n\rho^{\star}]=-1$. By $B$-symmetry, $\Tr[F_{AB_i}\rho^{\star}]=-1/n$ for every $i$. The global twirl makes each two-body marginal $U\otimes U$-invariant, so each marginal $\rho_{AB_i}$ is a Werner state with $f=-1/n$. Therefore
\[
p=\frac{1-2f}{3}=\frac{1+2/n}{3}=\frac{n+2}{3n}.
\]
This proves the bound is achievable. For any smaller $p$, mix this extremal extension with the fully mixed product extension. Hence symmetric extensions exist for all $0\le p\le (n+2)/3n$. Together with necessity, this proves Corollary~\ref{thm:werner-threshold}.

\section{Proof of the exact trace-distance resource}\label{sec:supp-trace-distance}

This section proves the $d=2$ case of Theorem~\ref{thm:Dm-d}.

The Werner twirl is a completely positive trace-preserving map. It leaves $\rho_W(p)$ invariant and maps any extendible state to an extendible Werner state. Since trace distance is contractive under CPTP maps, the closest extendible state to $\rho_W(p)$ can be chosen Werner. The optimization reduces to choosing $\sigma=\rho_W(q)$ with $q\le p_m^{\ast}$.

For two Werner states,
\[
\rho_W(p)-\rho_W(q)=(p-q)\Bigl(P_{-}-\frac{I}{4}\Bigr).
\]
The operator $P_{-}-I/4$ has eigenvalue $3/4$ on the singlet subspace and $-1/4$ on the three-dimensional triplet subspace. Therefore
\[
\bigl\|\rho_W(p)-\rho_W(q)\bigr\|_1=|p-q|\Bigl(\frac{3}{4}+3\cdot\frac{1}{4}\Bigr)=\frac{3}{2}|p-q|,
\]
and consequently $\tfrac12\|\rho_W(p)-\rho_W(q)\|_1=\tfrac34|p-q|$. The closest extendible Werner state has $q=\min\{p,p_m^{\ast}\}$. Thus
\[
D_m(p)=\frac{3}{4}\bigl(p-p_m^{\ast}\bigr)_+=\frac{3}{4}\Bigl(p-\frac{m+3}{3(m+1)}\Bigr)_+,
\]
which proves the $d=2$ case of Theorem~\ref{thm:Dm-d}.

\section{Tomographically complete behavior-level lower bound}\label{sec:supp-tomo-lower}

This section derives the conservative quantitative lower bound (superseded by Theorem~\ref{thm:exact-tomo}).

\paragraph{From state distance to behavior distance.} Let $T(\rho,\sigma)=\tfrac12\|\rho-\sigma\|_1$ denote the trace distance. For an $m$-colluder extension, let $\sigma_j=\rho_{AC_j}$ be the marginal between Alice and the $j$-th colluder. By symmetrizing the $B,C_1,\dots,C_m$ systems, the average state
\[
\tau=\frac{1}{m+1}\Bigl(\rho_W(p)+\sum_{j=1}^{m}\sigma_j\Bigr)
\]
is $m$-colluder extendible. Hence $D_m(p)\le T(\rho_W(p),\tau)$. By convexity of trace distance,
\[
T(\rho_W(p),\tau)\le\frac{1}{m+1}\sum_{j=1}^{m}T(\rho_W(p),\sigma_j).
\]
Therefore
\[
\frac{1}{m}\sum_{j=1}^{m}T(\rho_W(p),\sigma_j)\ge\frac{m+1}{m}D_m(p).
\]

\paragraph{Pauli-coefficient bound.} For any two-qubit Hermitian difference $\Delta=\rho-\sigma$, write the Pauli expansion
\[
\Delta=\frac{1}{4}\sum_{(\alpha,\beta)\neq(0,0)}r_{\alpha\beta}\,\sigma_\alpha\otimes\sigma_\beta,
\]
where $\sigma_0=I$ and $\sigma_{x,y,z}$ are the Pauli matrices. A conservative norm bound gives
\[
\|\Delta\|_1\le\sqrt{\sum_{(\alpha,\beta)\neq(0,0)}r_{\alpha\beta}^{2}}\le\sqrt{15}\,\max_{(\alpha,\beta)\neq(0,0)}|r_{\alpha\beta}|.
\]
Thus there exists at least one nontrivial Pauli coefficient with $|r_{\alpha\beta}|\ge\|\Delta\|_1/\sqrt{15}=2T(\rho,\sigma)/\sqrt{15}$. For the corresponding Pauli product measurement, the total-variation distance is at least $\tfrac12|r_{\alpha\beta}|\ge T(\rho,\sigma)/\sqrt{15}$.

If the capacity averages uniformly over the nine nontrivial product Pauli settings, this yields the conservative averaged bound
\[
\overline{\mathrm{TV}}_{\rm tomo}(\rho,\sigma)\ge\frac{T(\rho,\sigma)}{9\sqrt{15}}.
\]
Combining with the previous inequality and using the $d=2$ case of Theorem~\ref{thm:Dm-d} for $D_m(p)$ gives
\[
C_{m}^{\rm tomo}(p)\ge\frac{m+1}{9m\sqrt{15}}D_m(p)=\frac{m+1}{12m\sqrt{15}}\Bigl(p-\frac{m+3}{3(m+1)}\Bigr)_+.
\]
For one colluder this simplifies to
\[
C_{1}^{\rm tomo}(p)\ge\frac{1}{6\sqrt{15}}\Bigl(p-\frac{2}{3}\Bigr)_+.
\]
This bound is intentionally conservative, but it is rigorous and strictly positive throughout the source-non-shareable region.

\section{Exact tomographic capacity}\label{sec:supp-exact-tomo}

\begin{proof}
\emph{Step 1 (the colluder optimum is Werner).} The objective
$\rho_{AC}\mapsto d_{\mathrm{TV}}(P_{AB},\mathrm{beh}(\rho_{AC}))$ is convex,
and the feasible set $\{\rho_{AC}:\exists\,\rho_{ABC}\succeq0,\ \Tr_C\rho_{ABC}=\rho_W(p)\}$
is convex. Both are invariant under the octahedral rotation group
$O\subset\mathrm{U}(2)$ (the $24$ unitaries permuting $\{\pm X,\pm Y,\pm Z\}$)
acting as $U\otimes U\otimes U$: $\rho_W(p)$ is $O$-invariant, the nine-setting
uniform measure is permuted within itself, and $P_{AB}$ is fixed. Averaging
a feasible $\rho_{AC}$ over $O$ therefore does not increase the objective
(Jensen) and preserves feasibility, while sending the Bloch vectors to zero
and the correlation tensor to an isotropic multiple of the identity --- i.e.\
to a Werner state $\rho_W(q)$. Hence the optimum is attained on the Werner
line.

\emph{Step 2 (Werner colluder distance).} For $\rho_{AC}=\rho_W(q)$ only the
three aligned settings $x=y$ contribute, each with
$\tfrac12\sum_{a,c}|{-}p\,ac/4+q\,ac/4|=\tfrac12|p-q|$, while the six
$x\neq y$ settings give identical uniform statistics. Thus
$d_{\mathrm{TV}}(P_{AB},\mathrm{beh}(\rho_W(q)))=\tfrac19\cdot3\cdot\tfrac12|p-q|=\tfrac16|p-q|$.

\emph{Step 3 (closest joinable Werner).} Writing
$\Psi(p)=\Tr[V\rho_W(p)]=(1-3p)/2$, the pair $(\rho_W(p),\rho_W(q))$ is
joinable while sharing $A$ iff (Johnson--Viola, Cor.~III.3, $d=2$)
$\Psi(p)^2+\Psi(q)^2+\Psi(p)\Psi(q)\le\tfrac34$. For $p>\tfrac23$ one has
$\Psi(p)<-\tfrac12$, so the symmetric point $q=p$ is infeasible and the
closest feasible $q$ lies on the ellipse boundary, giving
$\Psi(q^\star)=\tfrac12\bigl(-\Psi(p)-\sqrt3\sqrt{1-\Psi(p)^2}\bigr)$ and
$|p-q^\star|=\tfrac12\bigl[(3p-1)-\sqrt{(3p+1)(1-p)}\,\bigr]$. Multiplying by
$\tfrac16$ from Step 2 yields the stated formula; $C^{\mathrm{tomo}}_1=0$ for
$p\le\tfrac23$ because there $q=p$ is joinable (perfect mimicry), recovering
the Theorem~\ref{thm:Dm-d} onset $p^*_1=\tfrac23$.
\end{proof}

\begin{suppremark}[Numerical certification and general $m$]
\label{rem:supp-tomo-sdp}
A direct semidefinite program minimising $d_{\mathrm{TV}}$ over all
tripartite extensions --- making no Werner assumption --- reproduces the
closed form to five decimals throughout $p\in[\tfrac23,1]$, confirming Step~1.
For $m$ colluders the capacity is the analogous SDP with a permutation-
symmetric extension over the $m$ colluder systems and the authorised
marginal fixed; its resource onset occurs exactly at the extendibility
threshold $p^{*}_m=(m+3)/[3(m+1)]$ (verified, e.g.\ $p^{*}_2=5/9$). The
one-colluder case is the one admitting the closed form above.
\end{suppremark}

\begin{suppremark}[What this fixes]
\label{rem:tomo-fix}
Theorem~\ref{thm:exact-tomo} removes the ``intentionally conservative''
constant $1/(6\sqrt{15})$ of the original Theorem~7, replacing a lower bound
by the exact capacity together with an explicit optimal colluder strategy
(the Werner state $\rho_W(q^\star)$). It also supplies the matching
achievability that the original Limitations paragraph flagged as missing in
the tomographic case.
\end{suppremark}

\section{Collective colluders and de~Finetti reduction}\label{sec:supp-definetti}

Theorem~\ref{thm:b2a} follows the standard de~Finetti route. Let $h^*$ be the optimal single-round witness and let $C>0$ be the single-shot capacity. Apply $h^*$ to $k$ tested rounds and average. Under the authorized behavior the rounds are i.i.d.\ with mean $\mu_P$, so Hoeffding gives $\Pr_P[\hat H\le\mu_P-C/2]\le\exp(-kC^2/2)$. Every round-marginal of the colluder lies in the single-round shadow, so $\mathbb E_{Q^{(n)}}[\hat H]\le\mu_P-C$ for every colluder. For permutation-symmetric $\rho^{(n)}$, the finite quantum de~Finetti theorem approximates the reduced state on $k$ tested rounds by a mixture of i.i.d.\ product states up to error $\delta=4D^2k/n$. Splitting the mixture into good and bad components, optimizing $k$ to balance exponential concentration against the de~Finetti overhead, yields $\beta_n=O(D^2\log n/(C^2n))$. Full details follow Christandl--K{\"o}nig--Mitchison--Renner~\cite{CKMR2007}, Renner~\cite{Renner2007}, Christandl--K{\"o}nig--Renner~\cite{CKR2009}, and the entropy-accumulation and randomness-extraction frameworks of Dupuis--Fawzi--Renner~\cite{Dupuis2020} and Pironio--Massar~\cite{Pironio2010}.

\begin{figure}[h]
\centering
\includegraphics[width=0.75\linewidth]{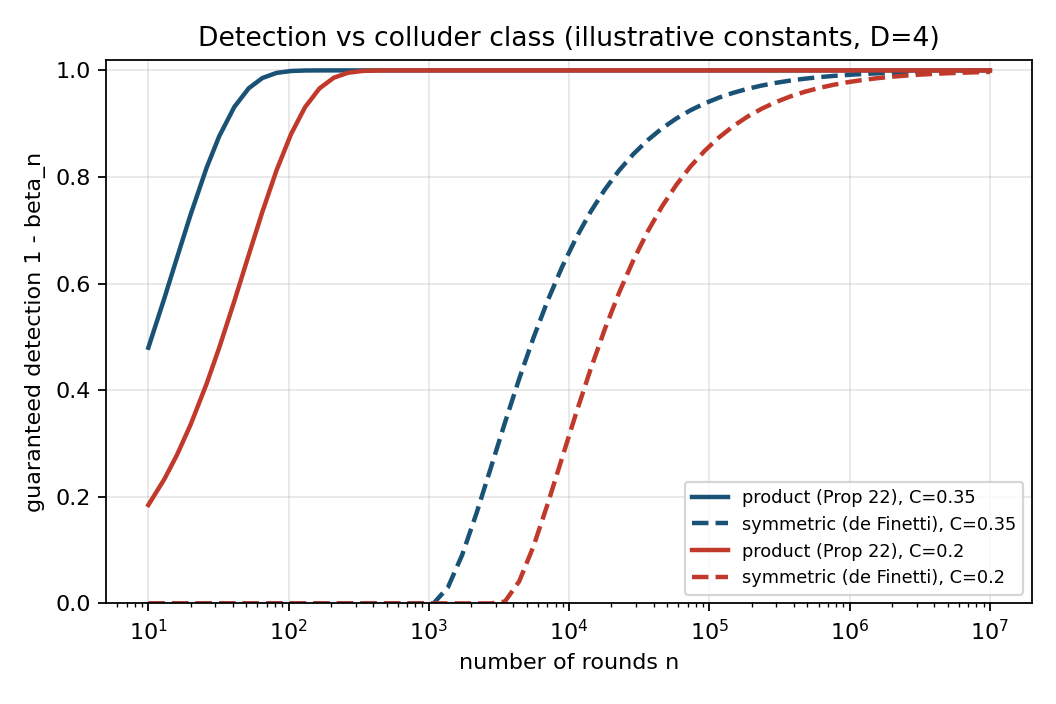}
\caption{\textbf{Detection advantage vs colluder class.} Soundness error $\beta_n$ as a function of the number of rounds $n$ for three colluder classes: product/memoryless (exponential decay, Proposition~\ref{prop:exp-detection}), collective/permutation-symmetric (polynomial decay $\tilde O(D^2\log n/(C^2n))$, Theorem~\ref{thm:b2a}), and arbitrary correlated (constant floor, Lemma~\ref{lem:multi-round-floor}). The de~Finetti overhead degrades the rate from exponential to polynomial, with the crossover around $n\sim10^4$--$10^5$ for the parameters shown.}
\label{fig:b2a-detection}
\end{figure}

\section{Collusion-resistant coordination: local model and SDP verification}\label{sec:supp-coordination}

%
%

\setcounter{supptheorem}{22}
Consider the game $G_{\mathrm{aligned}}$: a referee draws $i\in\{X,Y,Z\}$ uniformly and
sends the \emph{same} $i$ to Alice and to her partner; both measure $\sigma_i$ and score
$1$ iff their outcomes are anti-correlated. For a state $\rho$ the score is
$\mathcal S(\rho)=\tfrac13\sum_i\tfrac12\bigl(1-\langle\sigma_i\!\otimes\!\sigma_i\rangle_\rho\bigr)$.

\begin{suppproposition}[Explicit local model; vacuous device-independence]
\label{prop:supp-c1-local}
The aligned Werner behaviour $P_{AB}(a,b\mid i,j)=\tfrac14\bigl(1-p\,\delta_{ij}\,ab\bigr)$
is reproduced exactly, for every $p\in[0,1]$, by the hidden-variable model: draw
$s\in\{\pm1\}^3$ uniformly; with probability $p$ Alice outputs $s_i$ and the partner
outputs $-s_i$, otherwise Alice outputs $s_i$ and the partner outputs an independent
uniform bit. Consequently the behaviour is local for all $p$, and in the broadcast model
a colluder holding $s$ reproduces it verbatim, achieving the authorised score with zero
anti-collusion power. No device-independent certificate (which requires Bell violation)
yields any guarantee here.
\end{suppproposition}

\begin{supptheorem}[Source-certified coordination advantage]
\label{thm:supp-c1}
In $G_{\mathrm{aligned}}$ the authorised pair sharing $\rho_W(p)$ scores
$A(p)=\tfrac{1+p}{2}$, while every quantum-source colluder (any tripartite extension with
$\rho_{AB}=\rho_W(p)$) scores at most $V_Q(p)=\tfrac{1+q^\star(p)}{2}$, attained by the
Werner colluder $\rho_W(q^\star)$ with $q^\star$ the closest Johnson--Viola--joinable
weight. The source-certified advantage is
\begin{equation}
  \bigl[A(p)-V_Q(p)\bigr]_+
  =\tfrac14\Bigl[(3p-1)-\sqrt{(3p+1)(1-p)}\,\Bigr]\quad(p>\tfrac23),\qquad 0\ (p\le\tfrac23),
\end{equation}
i.e.\ $3\,C_1^{\mathrm{tomo}}(p)$. The advantage onsets exactly at the extendibility
threshold $p^*=\tfrac23$; for $p<\tfrac23$ one has $V_Q(p)>A(p)$ (the colluder
over-coordinates, the source being $2$-shareable). By Proposition~\ref{prop:supp-c1-local}
the broadcast/device-independent advantage is $0$ for all $p$.
\end{supptheorem}

\begin{proof}
$A(p)=\mathcal S(\rho_W(p))=\tfrac12(1-(-p))=\tfrac{1+p}{2}$. The colluder maximises
$\mathcal S(\rho_{AC})$ over $\rho_{AC}$ achievable with $\rho_{AB}=\rho_W(p)$ fixed; by
the octahedral symmetry of $G_{\mathrm{aligned}}$ and of the constraint the optimiser may
be taken Werner $\rho_W(q)$ (cf.\ Supplementary Note~\ref{sec:supp-exact-tomo}, Step~1), giving
$\mathcal S=\tfrac{1+q}{2}$, increasing in $q$. The largest joinable $q$ is $q^\star$ with
$p-q^\star=\tfrac12[(3p-1)-\sqrt{(3p+1)(1-p)}]$ (Johnson--Viola Cor.~III.3, as in
Thm.~\ref{thm:exact-tomo}), whence $A-V_Q=\tfrac{p-q^\star}{2}$. For $p\le\tfrac23$ the
symmetric point $q=p$ is joinable, so $V_Q\ge A$ and the advantage vanishes. A direct SDP
over all tripartite extensions confirms $V_Q(p)$ to four decimals and the sign change at
$p^*=\tfrac23$.
\end{proof}

\begin{suppcorollary}[Composition over rounds]
\label{cor:supp-c1-compose}
Over $n$ independent instances the advantage is maintained against collective
(permutation-symmetric) colluders with soundness error $\tilde O(D^2\log n/(C^2 n))$
(Theorem~\ref{thm:b2a}) and amplifies exponentially against memoryless colluders
(Proposition~\ref{prop:exp-detection}).
\end{suppcorollary}

\begin{suppremark}[Scope]
\label{rem:supp-c1-scope}
Theorem~\ref{thm:c1} is a collusion-resistant \emph{coordination} (anti-impersonation)
guarantee, not a cryptographic key or randomness primitive: the latter require min-entropy
bounds and a security rate valid against coherent, memory-bearing colluders (via entropy accumulation), which we
do not claim. The non-definitional content is Proposition~\ref{prop:c1-local}: the
advantage is attributable solely to the source assumption, and the device-independent
side is constructively shown to be empty on identical local statistics.
\end{suppremark}


\begin{thebibliography}{99}
\bibitem{Bell1964}
Bell, J. S. On the Einstein Podolsky Rosen paradox. \textit{Physics} \textbf{1}, 195--200 (1964).


\bibitem{CHSH1969}
Clauser, J. F., Horne, M. A., Shimony, A. \& Holt, R. A. Proposed experiment to test local hidden-variable theories. \textit{Phys. Rev. Lett.} \textbf{23}, 880--884 (1969).


\bibitem{Tsirelson1980}
Tsirelson, B. S. Quantum generalizations of Bell's inequality. \textit{Lett. Math. Phys.} \textbf{4}, 93--100 (1980).


\bibitem{Brunner2014}
Brunner, N., Cavalcanti, D., Pironio, S., Scarani, V. \& Wehner, S. Bell nonlocality. \textit{Rev. Mod. Phys.} \textbf{86}, 419--478 (2014).


\bibitem{Horodecki2009}
Horodecki, R., Horodecki, P., Horodecki, M. \& Horodecki, K. Quantum entanglement. \textit{Rev. Mod. Phys.} \textbf{81}, 865--942 (2009).


\bibitem{Barrett2005}
Barrett, J., Hardy, L. \& Kent, A. No signalling and quantum key distribution. \textit{Phys. Rev. Lett.} \textbf{95}, 010503 (2005).


\bibitem{ArnonFriedman2018}
Arnon-Friedman, R. et al. Practical device-independent quantum cryptography via entropy accumulation. \textit{Nat. Commun.} \textbf{9}, 459 (2018).


\bibitem{Portmann2022}
Portmann, C. \& Renner, R. Security in quantum cryptography. \textit{Rev. Mod. Phys.} \textbf{94}, 025008 (2022).


\bibitem{Pironio2009}
Pironio, S. et al. Device-independent quantum key distribution secure against collective attacks. \textit{New J. Phys.} \textbf{11}, 045021 (2009).


\bibitem{Acin2007}
Ac\'in, A. et al. Device-independent security of quantum cryptography against collective attacks. \textit{Phys. Rev. Lett.} \textbf{98}, 230501 (2007).


\bibitem{Masanes2011}
Masanes, L., Pironio, S. \& Ac\'in, A. Secure device-independent quantum key distribution with causally independent measurement devices. \textit{Nat. Commun.} \textbf{2}, 238 (2011).


\bibitem{Ekert1991}
Ekert, A. K. Quantum cryptography based on Bell's theorem. \textit{Phys. Rev. Lett.} \textbf{67}, 661--663 (1991).


\bibitem{Zapatero2023}
Zapatero, V. et al. Advances in device-independent quantum key distribution. \textit{npj Quantum Inf.} \textbf{9}, 10 (2023).


\bibitem{Nadlinger2022}
Nadlinger, D. P. et al. Experimental quantum key distribution certified by Bell's theorem. \textit{Nature} \textbf{607}, 682--686 (2022).


\bibitem{Zhang2022}
Zhang, W. et al. A device-independent quantum key distribution system for distant users. \textit{Nature} \textbf{607}, 687--691 (2022).


\bibitem{Liu2022}
Liu, W.-Z. et al. Toward a photonic demonstration of device-independent quantum key distribution. \textit{Phys. Rev. Lett.} \textbf{129}, 050502 (2022).


\bibitem{Wehner2018}
Wehner, S., Elkouss, D. \& Hanson, R. Quantum internet: A vision for the road ahead. \textit{Science} \textbf{362}, eaam9288 (2018).


\bibitem{Tavakoli2022}
Tavakoli, A., Pozas-Kerstjens, A., Luo, M.-X. \& Renou, M.-O. Bell nonlocality in networks. \textit{Rep. Prog. Phys.} \textbf{85}, 056001 (2022).


\bibitem{Terhal2004}
Terhal, B. M. Is entanglement monogamous? \textit{IBM J. Res. Dev.} \textbf{48}, 71--78 (2004).


\bibitem{Toner2006}
Toner, B. \& Verstraete, F. Monogamy of Bell correlations and Tsirelson's bound. \textit{Preprint at arXiv:quant-ph/0611001} (2006).


\bibitem{Coffman2000}
Coffman, V., Kundu, J. \& Wootters, W. K. Distributed entanglement. \textit{Phys. Rev. A} \textbf{61}, 052306 (2000).


\bibitem{Osborne2006}
Osborne, T. J. \& Verstraete, F. General monogamy inequality for bipartite qubit entanglement. \textit{Phys. Rev. Lett.} \textbf{96}, 220503 (2006).


\bibitem{Scarani2001}
Scarani, V. \& Gisin, N. Quantum communication between $N$ partners and Bell's inequalities. \textit{Phys. Rev. Lett.} \textbf{87}, 117901 (2001).


\bibitem{Aumann1974}
Aumann, R. J. Subjectivity and correlation in randomized strategies. \textit{J. Math. Econ.} \textbf{1}, 67--96 (1974).


\bibitem{Cleve2004}
Cleve, R., Hoyer, P., Toner, B. \& Watrous, J. Consequences and limits of nonlocal strategies. In \textit{Proc. 19th IEEE Conference on Computational Complexity}, 236--249 (IEEE, 2004).

\bibitem{Zhang2012}
Zhang, S. Quantum strategic game theory. In \textit{Proc. 3rd Innovations in Theoretical Computer Science Conference (ITCS '12)}, 39--59 (ACM, 2012).

\bibitem{Pappa2015}
Pappa, A. et al. Nonlocality and conflicting interest games. \textit{Phys. Rev. Lett.} \textbf{114}, 020401 (2015).

\bibitem{KDWW2019}
Kaur, E., Das, S., Wilde, M. M. \& Winter, A. Extendibility limits the performance of quantum processors. \textit{Phys. Rev. Lett.} \textbf{123}, 070502 (2019).

\bibitem{KDWW2021}
Kaur, E., Das, S., Wilde, M. M. \& Winter, A. Resource theory of unextendibility and nonasymptotic quantum capacity. \textit{Phys. Rev. A} \textbf{104}, 022401 (2021).

\bibitem{WangWilde2024}
Wang, K., Wang, X. \& Wilde, M. M. Quantifying the unextendibility of entanglement. \textit{New J. Phys.} \textbf{26}, 033013 (2024).

\bibitem{SinghWilde2025}
Singh, V. \& Wilde, M. M. Unextendible entanglement of quantum channels. \textit{IEEE Trans. Inf. Theory} \textbf{71}, 6002 (2025).

\bibitem{Werner1989}
Werner, R. F. Quantum states with Einstein-Podolsky-Rosen correlations admitting a hidden-variable model. \textit{Phys. Rev. A} \textbf{40}, 4277--4281 (1989).


\bibitem{NPA2007}
Navascu\'es, M., Pironio, S. \& Ac\'in, A. Bounding the set of quantum correlations. \textit{Phys. Rev. Lett.} \textbf{98}, 010401 (2007).


\bibitem{NPA2008}
Navascu\'es, M., Pironio, S. \& Ac\'in, A. A convergent hierarchy of semidefinite programs characterizing the set of quantum correlations. \textit{New J. Phys.} \textbf{10}, 073013 (2008).


\bibitem{DPS2002}
Doherty, A. C., Parrilo, P. A. \& Spedalieri, F. M. Distinguishing separable and entangled states. \textit{Phys. Rev. Lett.} \textbf{88}, 187904 (2002).


\bibitem{DPS2004}
Doherty, A. C., Parrilo, P. A. \& Spedalieri, F. M. Complete family of separability criteria. \textit{Phys. Rev. A} \textbf{69}, 022308 (2004).


\bibitem{Kaniewski2016}
Kaniewski, J. Analytic and nearly optimal self-testing bounds for the cloned-CHSH and Bell inequalities. \textit{Phys. Rev. Lett.} \textbf{117}, 070402 (2016).


\bibitem{Acin2012}
Ac\'in, A., Massar, S. \& Pironio, S. Randomness versus nonlocality and entanglement. \textit{Phys. Rev. Lett.} \textbf{108}, 100402 (2012).


\bibitem{Bamps2015}
Bamps, C. \& Pironio, S. Sum-of-squares decompositions for a family of Clauser-Horne-Shimony-Holt-like inequalities and their application to self-testing. \textit{Phys. Rev. A} \textbf{91}, 052111 (2015).


\bibitem{Yang2013}
Yang, T. H. \& Navascu\'es, M. Robust self-testing of unknown quantum systems into any entangled two-qubit states. \textit{Phys. Rev. A} \textbf{87}, 050102(R) (2013).


\bibitem{Uhlmann1976}
Uhlmann, A. The transition probability in the state space of a $*$-algebra. \textit{Rep. Math. Phys.} \textbf{9}, 273--279 (1976).


\bibitem{Fuchs1999}
Fuchs, C. A. \& van de Graaf, J. Cryptographic distinguishability measures for quantum-mechanical states. \textit{IEEE Trans. Inf. Theory} \textbf{45}, 1216--1227 (1999).


\bibitem{Dupuis2020}
Dupuis, F., Fawzi, O. \& Renner, R. Entropy accumulation. \textit{Commun. Math. Phys.} \textbf{379}, 867--913 (2020).


\bibitem{Pironio2010}
Pironio, S. \& Massar, S. Security of practical private randomness generation. \textit{Phys. Rev. A} \textbf{87}, 012336 (2013).

\bibitem{Wootters1982}
Wootters, W. K. \& Zurek, W. H. A single quantum cannot be cloned. \textit{Nature} \textbf{299}, 802--803 (1982).

\bibitem{JV2013}
Johnson, P. D. \& Viola, L. Compatible quantum correlations: Extension problems for Werner and isotropic states. \textit{Phys. Rev. A} \textbf{88}, 032323 (2013).

\bibitem{JSZ2022}
Jakab, D., Solymos, A. \& Zimbor\'as, Z. Extendibility of Werner states. \textit{arXiv preprint} arXiv:2208.13743 (2022).

\bibitem{CKMR2007}
Christandl, M., K\"onig, R., Mitchison, G. \& Renner, R. One-and-a-half quantum de Finetti theorems. \textit{Commun. Math. Phys.} \textbf{273}, 473--498 (2007).

\bibitem{Renner2007}
Renner, R. Symmetry of large physical systems implies independence of subsystems. \textit{Nat. Phys.} \textbf{3}, 645--649 (2007).

\bibitem{CKR2009}
Christandl, M., K\"onig, R. \& Renner, R. Postselection technique for quantum channels with applications to quantum cryptography. \textit{Phys. Rev. Lett.} \textbf{102}, 020504 (2009).



\clearpage
\end{thebibliography}
\end{document}